\documentclass{tlp}
\usepackage{amsmath, amssymb, bbm}
\usepackage[utf8]{inputenc}

\usepackage[scaled=0.8]{beramono}
\usepackage[T1]{fontenc}
\usepackage{microtype}
\usepackage{csquotes}
\usepackage{needspace}

\usepackage[hyphens]{url}
\usepackage{hyperref}
\usepackage[capitalize]{cleveref}

\usepackage{graphicx}
\usepackage{tikz}
\usetikzlibrary{arrows,positioning,shapes.misc,calc}
\definecolor{lightgrey}{RGB}{170, 170, 170}
\pgfdeclarelayer{background}
\pgfdeclarelayer{foreground}
\pgfsetlayers{background,main,foreground}
\usepackage{pgf-umlsd}
\usepackage{caption}
\usepackage{subcaption}

\usepackage{fancyvrb}
\usepackage{relsize}
\usepackage{listings}
\usepackage{verbatim}
\usepackage{alltt}
\newcommand\listingsize{\fontsize{9pt}{10pt}}
\RecustomVerbatimCommand{\Verb}{Verb}{fontsize=\listingsize}
\RecustomVerbatimEnvironment{Verbatim}{Verbatim}{fontsize=\listingsize}
\crefname{lstlisting}{Listing}{Listings}
\definecolor{grey}{RGB}{130,130,130}

\usepackage{xspace}

\newcommand{\dbpedia}{DBpedia\xspace}
\newcommand{\cwm}{{\itshape cwm}\xspace}
\newcommand{\http}{{HTTP}\xspace}

\newcommand{\nthree}{{N\oldstylenums 3}\xspace}

\newcommand{\owls}{\mbox{OWL-S}\xspace}
\newcommand{\rdf}{RDF\xspace}

\newcommand{\restdesc}{RESTdesc\xspace}

\newcommand{\uri}{URI\xspace}

\newcommand{\URL}{URL\xspace}

\newtheorem{theorem}{Theorem}
\newtheorem{definition}[theorem]{Definition}
\newtheorem{remark}[theorem]{Remark}
\newtheorem{lemma}[theorem]{Lemma}
\newtheorem{corollary}[theorem]{Corollary}

\submitted{August 15, 2013}
\revised{December 18, 2013, November 13, 2014, August 31, 2015}
\accepted{December 9, 2015}

\begin{document}

\title[The Pragmatic Proof:    Hypermedia API Composition and Execution]{%
       The Pragmatic Proof:\\* Hypermedia API Composition and Execution
}
\author[Ruben Verborgh et al.]
{
  RUBEN VERBORGH, D\"ORTHE ARNDT, SOFIE VAN HOECKE\\
  Ghent University -- iMinds -- Multimedia Lab\\
  Sint-Pietersnieuwstraat 41,\\
  9000 Ghent, Belgium
  \and
  JOS DE ROO and GIOVANNI MELS\\
  Agfa Healthcare\\
  Moutstraat 100,\\
  9000 Ghent, Belgium
  \and
  THOMAS STEINER and JOAQUIM GABARRO\\
  Universitat Polit\`ecnica de Catalunya\\
  Department {\scriptsize LSI}\\
  08034 Barcelona, Spain
}

\maketitle

\begin{abstract}
Machine clients are increasingly making use of the Web to perform tasks.
While Web services traditionally mimic remote procedure calling interfaces,
a~new generation of so-called hypermedia APIs works through hyperlinks and forms,
in a~way similar to how people browse the Web.
This means that existing composition techniques,
which determine a~procedural plan upfront,
are not sufficient to consume hypermedia APIs,
which need to be navigated at runtime.
Clients instead need a~more dynamic plan
that allows them to follow hyperlinks and use forms with a~preset goal.
Therefore, in this article,
we show how compositions of hypermedia~APIs can be created by generic Semantic Web reasoners.
This is achieved through the generation of a~proof
based on semantic descriptions of the APIs' functionality.
To pragmatically verify the applicability of compositions,
we introduce the notion of \mbox{pre-execution} and \mbox{post-execution proofs}.
The runtime interaction between a~client and a~server is guided by proofs
but driven by hypermedia,
allowing the client to react to the application's actual state indicated by the server's response.
We~describe how to generate compositions from descriptions,
discuss a~computer-assisted process to generate descriptions,
and verify reasoner performance on various composition tasks
using a~benchmark suite.
The experimental results lead to the conclusion that proof-based consumption of hypermedia APIs
is a~feasible strategy at Web scale.
\\

Under consideration in Theory and Practice of
Logic Programming (TPLP).
\end{abstract}


\maketitle

\clearpage

\section{Introduction}

\subsection{Hard-coded API contracts on a~hypermedia-driven Web}
The World Wide Web,
on which millions of servers together offer billions of information resources,
has been designed as a~distributed hypermedia application~\cite{bernerslee_1992}.
\enquote{Hypermedia} means that pieces of information can be connected to each other;
therefore, consuming information
does not require any knowledge of servers' internal information structures.
Instead, users of the Web follow \emph{hyperlinks} and fill out \emph{forms}
to move from one piece of information to the next.
This architectural decision has been essential for the global growth of the Web:
people can browse websites by clicking links,
regardless of whether they have used them before.
Fielding called this principle
\emph{\enquote{hypermedia as the engine of application state}} \cite{REST},
because servers send hypermedia documents,
which clients use to advance the state of the interaction.
Rather than relying on a~pre-determined set of actions
that would have been communicated through a~separate channel,
such hypermedia documents let clients select steps just-in-time.
This ensures the Web's temporal scalability:
if a~server decides to change its interface,
clients do not have to be reprogrammed---%
they simply receive hypermedia documents with different~links.

As more and more people found their way to the Web,
it seemed evident that automated clients would also start using the Web autonomously.
The immediate barrier was that all resources on the Web
were only available in human languages,
which machines cannot accurately interpret yet.
A~logical step for programmers,
who deal with Application Programming Interfaces (APIs) in software development,
was to retrofit a~system for API operations
to the Web's \http interface~\cite{HTTP}.
Instead of navigating links without prior knowledge,
machines executed a~pre-defined list of commands
that they translated into \http requests.
As such, those APIs follow a~proprietary Remote Procedure Calling (RPC) protocol
\emph{on~top~of} \http,
which---despite the label \enquote{Web~services} or \enquote{Web~{APIs}}---%
has consequently few in common with the~Web.

As expected, such proprietary APIs with a~fixed contract
cannot withstand evolution very well.
If the server changes the API,
clients have to be reprogrammed.
While incompatible API changes are uncommon in closed environments,
servers on the Web are in constant evolution,
as witnessed by the short lifespan of websites and APIs.
The hypermedia mechanism that guides clients through the Web---%
and thus allows them to cope with changes---%
is notably absent from RPC APIs.

\subsection{Hypermedia APIs as native Web citizens}
\label{sec:RESTConstraints}
A~decade after the invention of the Web,
Fielding analyzed the architectural principles that contributed to its world-wide growth,
which he captured in the Representational State Transfer (REST) architectural style~\cite{REST}.
Unlike the RPC APIs discussed above,
APIs that follow the REST principles are native Web citizens,
and thus more resilient to change~\cite{verborgh_jod_2014}.
The distinguishing characteristic of REST is its \emph{uniform interface},
consisting of four constraints~\cite{REST}:

\begin{description}
  \item[Identification of resources]
  The essential unit of information in REST architectures is a~\emph{resource},
  a~conceptual entity that must be \emph{uniquely identifiable}.
  On the Web, this means that each resource must have its own URL.
  For example, \verb!http://example.org/weather/london/2014/05/01!
  could identify the weather in London on a~particular day.

  \item[Resource manipulation through representations]
  Resources are conceptual and cannot be transfered;
  clients and servers instead exchange \emph{representations}.
  Depending on the client's capabilities,
  client and server agree on one of multiple possible representations of each resource
  (e.g., HTML for humans, JSON or RDF for machines).
  For example, the weather in London with the above URL
  would be represented as a~JSON document
  when requested by a~JavaScript application.

  \item[Self-descriptive messages]
  Rather than defining custom actions,
  as is the case with APIs in typical programming languages,
  REST APIs use a~limited set of commands defined by a~protocol.
  For example, the Web uses HTTP's \verb!GET! and \verb!PUT!,
  which have a universal meaning,
  rather than \verb!showWeather! or \verb!setWeather!.

  \item[Hypermedia as the engine of application state]
  Also known as the \emph{hypermedia constraint},
  this principle indicates that the client should be able
  to perform the interaction with the server solely through hypermedia.
  On the Web, this happens through \emph{hypermedia controls}
  such as hyperlinks and forms.
  Another perspective on this is that all communication should happen \emph{in-band}
  instead of \emph{out-of-band}:
  clients engage in an interaction through hypermedia representations of resources
  rather than through pre-defined contracts.
  REST APIs are thus \emph{hypermedia-driven}~\cite{fielding_2008}.
  For example, tomorrow's weather is linked from today's weather,
  rather than having to craft a~new weather request by hand.
\end{description}

Unfortunately, many APIs mistakenly label themselves as \enquote{REST},
giving a~rather unclear meaning to the term \enquote{REST API};
in particular, the fourth constraint is often missing.
As a~result, the term \enquote{\emph{hypermedia API}} is used
to distinguish those APIs that follow all REST constraints~\cite{RESTfulWebApis},
and thus inherit their architectural benefits.

An important consequence of the REST architectural constraints
is that there is no observable difference
between \emph{websites} and \emph{hypermedia APIs}.
The term \enquote{website} is commonly applied when the consumers are humans,
and \enquote{hypermedia~API} is used for machine clients.
However, they only differ in the kind of representations they offer:
human-readable or machine-readable.
By offering multiple representations,
a~\emph{single} hypermedia API/website can serve
both humans and machine clients,
as opposed to the two separated interfaces we commonly see~\cite{verborgh_jod_2014}.

\subsection{Generic clients of hypermedia APIs}\label{API_comp}
Once APIs have been made accessible for machines,
the question becomes:
how can we create clients that perform tasks using such APIs?
And more specifically,
can we build generic clients that do not have to be preprogrammed
for a~specific task?
For example, given an API for shipping packages,
an API for user profiles, and an API of online bookstore,
how can we deliver a~book on a~user's wish~list to their doorstep?
Indeed, in addition to interpreting an API's responses,
clients should be able to reason about an APIs' functionality
and how it can be combined with other APIs to perform a~complex task.
This question has been the subject of much literature
for classical Web services,
where the APIs were composed into a~static plan
that had to be executed step-by-step.
Such a~plan, and the mechanism to generate it, would only be valid
as long as all involved APIs did not undergo any changes,
which is a~rather unrealistic assumption on the Web.

The situation is fundamentally different for hypermedia APIs:
because of the hypermedia constraint,
no such static plan can (nor should) be created beforehand,
because hypermedia APIs can only be consumed by following hypermedia controls.
On the other hand, clients should somehow know
what links they must follow,
because only certain links will lead them
to successful completion of the desired~goal.
This gives rise to an approach where we need a~high-level plan
that guides the runtime, hypermedia-driven execution of Web APIs.

In this article, we introduce and formalize a~proof-based method
to produce such high-level plans,
and detail their execution through hypermedia.
In the next section,
we describe related work in the domain of Web services and Web APIs,
followed by a~justification of the chose technologies in \cref{sec:Technologies}.
We introduce to hypermedia API description in \cref{sec:RESTdesc}.
\cref{sec:Proof} explains the use of proofs
to validate hypermedia API compositions,
the generation of which is detailed in \cref{sec:Composition}.
We explain how descriptions can be created in a~computer-assisted process
in \cref{sec:Generation}.
The feasibility of the approach is evaluated in \cref{sec:Evaluation}.
Finally, we end the article in \cref{sec:Conclusion} with conclusions
and an outlook on future~work.

\section{Related Work}
\label{sec:RelatedWork}

\subsection{Web APIs}
Let us start a~discussion on related work by defining terms about Web APIs.

\begin{description}\label{api}
  \item[A~Web server] uses the HTTP~protocol
    to offer data and/or actions,
    which are often exclusively located on this server.
    In other words, the server performs a~data lookup or computation
    that another device cannot or does not.
  \item[A~Web client] consumes resources offered by a~Web~server.
  \item[A~Web service or Web API] is a~machine-accessible interface
    to data and/or actions offered by a~server through HTTP.
    Traditionally, \enquote{Web service} is used for RPC-based XML interfaces,
    whereas \enquote{Web API} is used for more lightweight APIs,
    such as those based on JSON.
  \item[A~Web API operation]
    is a~single HTTP request from a~Web client to a~Web API.
    Interactions between a~client and a~server consist of one or more operations.
  \item[A~hypermedia control] is a~hyperlink or form
    that allows clients to navigate from one resource to another.
  \item[A~hypermedia API] is a~Web API
    that follows all REST architectural constraints
    as discussed in \cref{sec:RESTConstraints}.
    In particular, it is designed such that its (human or machine) clients
    should perform the interaction through hypermedia controls.
\end{description}

A~characteristic of hypermedia APIs is that they,
in contrast to RPC~APIs,
are not action-based but resource-based.
As such, the border between data and actions is blurred;
they are effectively modeled as one and the same.
For example, an RPC~API might offer access to image resizing functionality
by adding a~dedicated action named \verb!resizeImage!.
A~hypermedia API would instead expose a~resource for the original image
such as \verb!/images/381/!,
which would allow access to resized images
through a~links towards the resource \verb!/images/381/thumbnail/!.
This structuring around resources
considerably simplifies the planning process,
as actions do not have to be considered as separate entities.
Consequently, as we will see in \cref{sec:RESTdesc},
rules and data need not be treated differently.

Integrating hypermedia APIs into an application nowadays requires manual development work,
such as writing the \http request templates
and parsing the returned \http responses.
Instructions on how to write this code can often be found
on the API's website in the form of human-readable API documentation.
In order to automate this process, \emph{machine-readable documentation} is necessary.
On the lowest level, this documentation describes the message format and modalities.
On a~higher level, it explains the specific functionality offered by the hypermedia API,
so a~machine can autonomously decide whether the API is appropriate for a~certain use case.

\subsection{The Semantic Web}
To create machine-readable information in general,
we can use standards from the Semantic Web~\cite{SemanticWeb},
which is a~vision that enables autonomous clients to use the Web.
A~central standard is \rdf~\cite{RDF},
which is a~machine-interpretable language consisting of data triples
$(\textit{subject}, \textit{predicate}, \textit{object})$.
For instance, the following triple describes the relationship
between Leonhard Euler and Daniel Bernoulli:
\begin{Verbatim}
  <http://dbpedia.org/resource/Leonhard_Euler>
      <http://xmlns.com/foaf/0.1/knows>
          <http://dbpedia.org/resource/Daniel_Bernoulli>.
\end{Verbatim}
Various RDF syntaxes exist;
some of them, such as Notation3 (\nthree,~\citeNP{Notation3})
extend the RDF model with features such as variables and quantification.
The proof-based algorithm in this paper
uses the \nthree reasoner EYE~\cite{eyepaper} to generate plans.
We chose N3 over other formalisms such as Prolog~\cite{Prolog}
or Datalog~\cite{datalog}
because RDF---and thus N3---is native to the Web.
This is exemplified in the triple above:
the predicate that relates Euler to Bernoulli
can be seen as a~(typed) hyperlink from the first URL to the other.
Thereby, the hyperlink concept that is crucial for the Web and hypermedia APIs,
can be represented in the most straightforward way in RDF/N3.
Furthermore, the presence of RDF means that we can reuse ontological constructs
from well-known vocabularies such as RDFS~\cite{RDFS} or OWL~\cite{OWL}.
For example, the domain and range of the \verb!foaf:knows! predicate
are \verb!foaf:Person!.
Hence, using the triple above and this ontological knowledge,
additional triples can be derived,
expressing that Euler and Bernoulli are instances of \verb!foaf:Person!.
Such derived knowledge could then be used
for hypermedia API operations
that require \verb!foaf:Person! instances as input.

\subsection{Web Service Description}
\label{subsec:WebServiceDescription}
Machine-readable documentation of Web services
has been a~topic of intense research
for at least a~decade.
There are many approaches to service description
with different underlying service models.
\owls~\cite{OWLS} and WSMO~\cite{WSMO} are the most well-known
Semantic Web Service description paradigms.
They both allow to describe the high-level semantics of services
whose message format is WSDL~\cite{WSDL}.
Though extension to other message formats is possible,
this is rarely seen in practice.
Semantic Annotations for WSDL~(SAWSDL, \citeNP{SAWSDL})
aim to provide a~more lightweight approach for bringing semantics to WSDL services.
Composition of Semantic Web services has been well documented,
but all approaches focus on RPC interactions and require specific software~\cite{CurrentComposition}.
In contrast, the proposed approach works for REST interactions and exclusively
relies on \emph{generic} Semantic Web reasoners.
While automated approaches to create descriptions are being researched~\cite{Leandro},
they are not the focus of this paper.

\subsection{Web API Description}
\label{subsec:WebAPIDescription}
In recent years, several description formats
for the more lightweight Web APIs have emerged~\cite{verborgh_rest_2014}.
Linked Open Services~(LOS, \citeNP{LOS}) expose functionality on the Web
using Linked Data technologies, such as HTTP and RDF.
Input and output parameters are described with graph patterns embedded inside RDF string literals to achieve quantification, which RDF does not support natively.
Linked Data Services~(LIDS, \citeNP{LIDS}) define interface conventions supported by a~lightweight model.
None of these methods, however, use the hypermedia principles of REST.
Several methods aim to enhance existing technologies
to deliver annotations of Web APIs.
HTML for RESTful Services~(hRESTS, \citeNP{hRESTS})
is a~microformats extension
to annotate HTML descriptions of Web APIs in a~machine-processable way.
SA-REST~\cite{SAREST} provides an extension of hRESTS
that describes other facets such as data formats and programming language bindings.
MicroWSMO~\cite{MicroWSMO,Maleshkova2009},
an extension to SAWSDL that enables the annotation of RESTful services,
supports the discovery, composition, and invocation of Web APIs,
but requires additional software.
Data-Fu~\cite{DataFu} uses rules to describe client-server interactions;
however, these rules are tightly bound to a~server's information structure
and are equivalent to a~fixed declarative program in RPC style.
This paper instead introduces a~flexible, hypermedia-driven technique
for the REST architectural style.

\subsection{Hypermedia API Description}
The description of hypermedia APIs is a relatively new field.
Hydra~\cite{HydraVocabulary} is a~vocabulary to support API descriptions,
but does not directly support automated composition.
\restdesc~\cite{verborgh_wsrest_2012}
is a~description format
for hypermedia APIs that describes them in terms of resources and~links.
\restdesc is expressed in \nthree
and will be used as a~description format in this paper,
so it is discussed further in \cref{sec:RESTdesc}.
The Resource Linking Language (ReLL, \citeNP{ReLL}) features media types, resource types, and link types as first-class citizens for descriptions.
The authors of ReLL also propose a~method
for ReLL API composition~\cite{ReLLComposition}
using Petri nets to describe the machine-client navigation.
However, in contrast to \restdesc, it does not support
automatic, functionality-based composition.

\subsection{Semantic Web Reasoning}
\label{subsec:RelatedReasoning}
The Pellet reasoner~\cite{Pellet} and the various reasoners of the Jena framework~\cite{Jena}
are the most commonly known examples of publicly available Semantic Web reasoners.
Pellet is a~reasoner on ontological constructs~\cite{OWL},
while Jena offers various ontological and rule-based reasoners.
The rule reasoner is the most flexible,
as it allows to incorporate custom derivations,
but it uses a~rule language that is specific to Jena and therefore not interchangeable.

Another category of reasoners uses \nthree,
leveraging the language's support of formulas and quantification for RDF
to provide a~logical framework for inferencing~\cite{N3Logic}.
The first \nthree reasoner was the forward-chaining \cwm~\cite{cwm},
which is a~general-purpose data processing tool for \rdf,
including tasks such as querying and proof-checking.
Another important \nthree reasoner is EYE~\cite{eyepaper},
whose features include backward-chaining and high performance.
A~useful capability of both \nthree reasoners
is their ability to generate and exchange \emph{proofs},
which can be used for software synthesis or API composition~\cite{Manna,Waldinger}.

\section{Justification of Chosen Technologies}
\label{sec:Technologies}
In this section, we explain and justify the technological choices made in this article.
First, we detail our choice for REST Web APIs, emphasizing the fundamental differences with RPC~APIs.
Second, we argue why we choose Notation3 instead of other logic programming languages.

\subsection{REST Web APIs}
\label{subsec:TechnologiesREST}
With the exception of RESTdesc,
existing techniques for Web service and Web API description
(discussed in \cref{subsec:WebServiceDescription,subsec:WebAPIDescription})
focus exclusively on interactions that follow the RPC model.
Algorithms that create compositions of such services or APIs
will essentially produce a~list of calls that have to be issued.
This process is visualized in \cref{fig:RPC}.
Two things are of particular interest in this diagram:
\begin{itemize}
  \item While subsequent calls might use output from earlier tasks
        (and might even be conditional based on such output),
        the type of calls made is not influenced by the RPC~API.
        The \emph{control flow} is dictated by the composition,
        based on descriptions of each call.
        The API itself does not provide any information
        about which next steps can be taken
        and how these should be performed.
        The role of descriptions is thus three-fold:
        a)~explaining which order of calls are possible;
        b)~describing how to perform calls;
        c)~expressing the functionality of~calls.
  \item Consequently, all control information is in the descriptions and composition;
        the server does not send any control information.
        The client thus fully relies on the descriptions and composition for control information.
\end{itemize}

We conclude that, apart from their label, Web services or APIs that follow the RPC communication style
have only a~minor connection to the Web.
While technically, they tunnel their calls over the HTTP protocol,
nothing in their principled workings is tied to the Web's core principles,
which include hypermedia documents and~hyperlinks.

\begin{figure}[t]
  \begin{subfigure}{0.46\linewidth}
    \includegraphics[width=\linewidth]{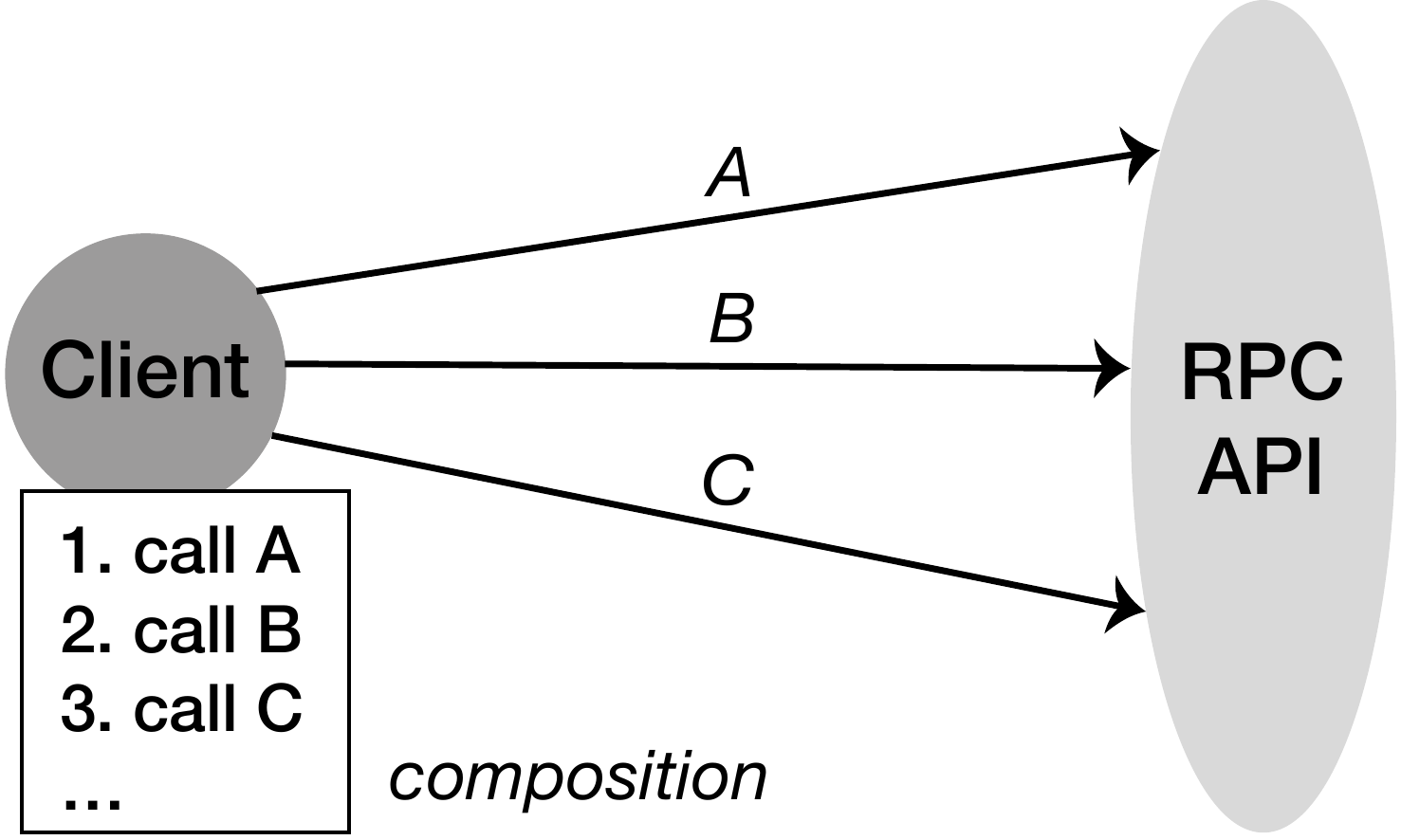}
    \caption{RPC composition and interaction}
    \label{fig:RPC}
  \end{subfigure}
  \hfill
  \begin{subfigure}{0.46\linewidth}
    \includegraphics[width=\linewidth]{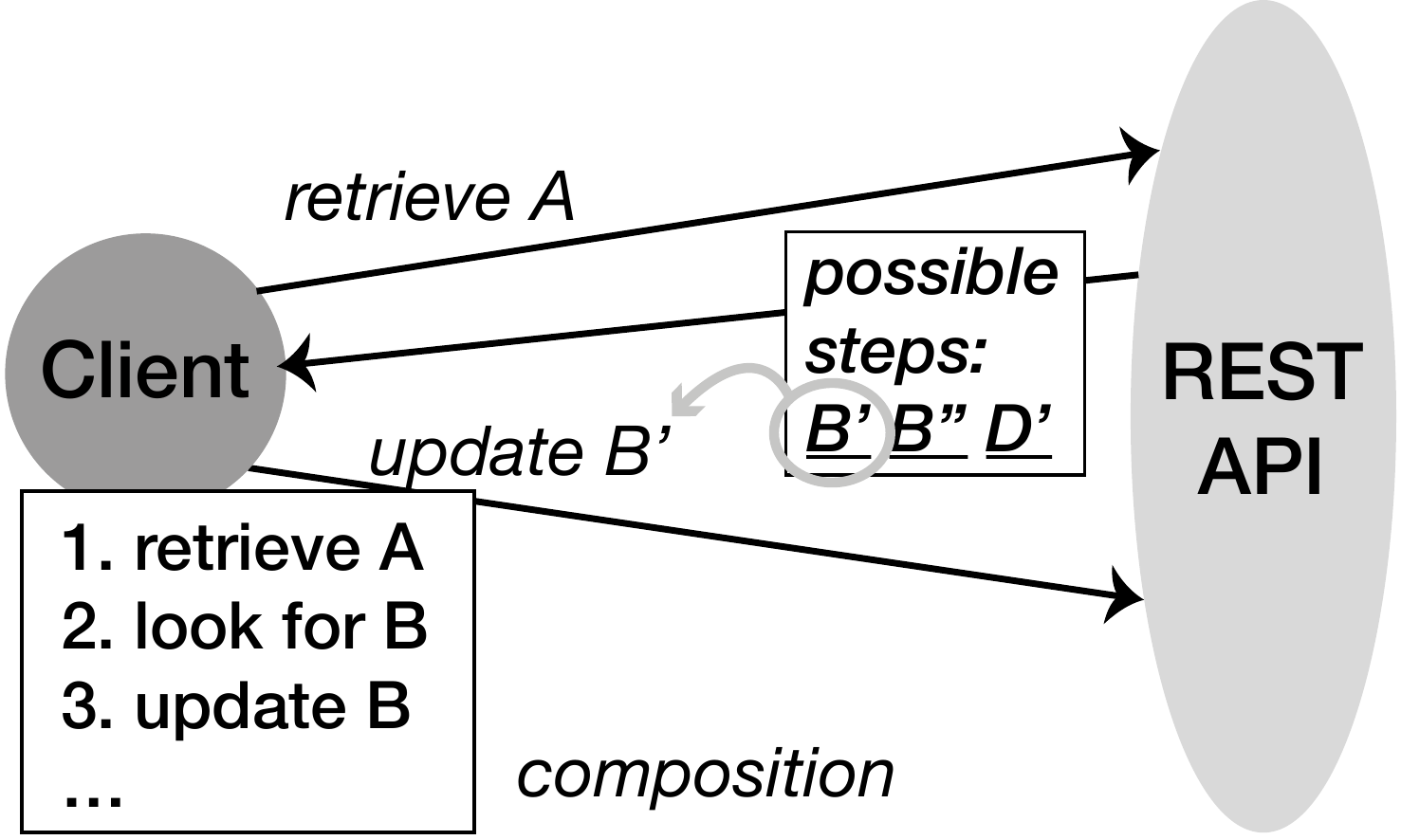}
    \caption{REST composition and interaction}
    \label{fig:REST}
  \end{subfigure}
  \caption{Interaction differences between clients of RPC and REST APIs}
  \label{fig:RPCvsREST}
\end{figure}

In contrast, an API that follows the REST architectural style
provides (information) \emph{resources} rather than (procedural) \emph{calls} as its interface.
It replies to a~request for a~resource with a~\emph{hypermedia representation},
which contains hyperlinks and forms that provide access to possible next steps.
\Cref{fig:REST} illustrates this principle:
a~client first requests a~resource~$A$,
which results in a~hypermedia document.
This documents contains a~link to~$B'$,
which is a~concrete instantiation of the~$B$
the client should look for according to the composition.
Using the control information in the hypermedia document,
the client then performs an action on~$B'$.
This indicates two major differences compared to~RPC:
\begin{itemize}
  \item The REST control flow is not fully dictated by the composition.
        The composition only provides a~high-level plan of what needs to happen,
        whereas concrete actions are performed through hypermedia controls.
        Descriptions have only one function,
        namely expressing the functionality provided by a~certain resource type.
        The access order or an explanation of how the call is executed
        are determined at runtime based on the REST API's responses.
  \item As such, the client needs to derive a~control flow at runtime
        by combining high-level steps from a~composition
        with concrete actions from hypermedia documents obtained from the server.
\end{itemize}

Note how REST APIs do have a~strong connection to the Web's core principles,
as clients use hypermedia documents and hyperlinks to perform tasks.
Unlike the RPC style, this is similar to human consumption of the Web:
when we want to perform a~certain task, we also have a~high-level plan in mind
(e.g., ordering a~package involves entering a~delivery address)
that becomes concrete through hypermedia controls
(e.g., using a~form to submit that delivery address).
This allows for a~much looser coupling~\cite{WebLooselyCoupled},
as descriptions and compositions do not need to contain interaction details.

Furthermore, because of the emphasis on links,
REST APIs are closely related to RDF.
If an API outputs in an RDF-based format (such as Turtle, JSON-LD, or N3),
the client obtains a~list of triples,
which are essentially typed links (if we interpret the property as a~relation specifier).
Following these links to realize a~concrete action
is an example of hypermedia-driven Web API consumption,
which is what we will outline in this article.
We stress that this is not possible with existing descriptions and algorithms
for RPC Web services and Web APIs,
as they cannot incorporate runtime control information originating from the API.

\subsection{Notation3}
\label{subsec:TechnologiesN3}
In this article, we use the Notation3 (\nthree) language and associated logic framework
to generate and execute compositions of REST Web APIs.
We justify this choice on both a~practical level
and a~theoretical level, as detailed in the subsections below.

\subsubsection{Practical Arguments}
The Linked Data Cloud contains billions of triples in the RDF model~\cite{LODCloud}.
In order for this knowledge to be reused,
we need a~logic with support for triples.
Whereas typical programming languages can all support triples in some way,
treating them as first-class citizens decreases the effort of working with them.
For instance, while the popular Jena library provides triple support for Java,
an RDF document is not a~Java document,
so the RDF will need to be converted in Java-specific objects.
Programmatic access and manipulation happens through these objects
instead of directly at the RDF level.
In contrast, \nthree is a~superset of the Turtle serialization of RDF.
This means that each valid RDF document, when expressed in Turtle,
is by definition a~valid \nthree document.
As such, reasoners that support \nthree natively support triples,
without the impedance mismatch of other languages.
This means that all RDF responses on the Web,
including those of RDF Web~APIs,
can be interpreted directly by \nthree reasoners.
Usage of \nthree thus brings direct compatibility
with a~body of billions of knowledge facts on the Web.

Additionally, hundreds of ontologies are expressed as RDF triples using RDFS and OWL.
That means they use RDF triples with RDFS and OWL predicates and objects
to define ontological meaning.
Because they are expressed as triples, \nthree reasoners can process them.
However, since many common RDFS and OWL constructs
have been captured in publicly available \nthree rules~\cite{EYE},
\nthree reasoners can also apply their semantics
without requiring native support for neither RDFS nor~OWL.
For example, the OWL class \verb!owl:SymmetricProperty!
is implemented by stating that for properties~$p$ that are symmetric,
the existence of a~triple $(s, p, o)$ implies the existence of $(o, p, s)$.
Therefore, by using \nthree, we can directly incorporate available ontological~knowledge.

Finally, the most important reason to choose \nthree as an underlying description format
is that it can combine the above aspects together with descriptions of Web~APIs.
As argued above, the hypermedia-driven nature of REST Web APIs maps well to the RDF model.
If we use an \nthree rule to describe Web~APIs,
we can directly reason on the combination of Web~API descriptions,
existing triples, and existing ontologies.
Concretely, if a~Web API returns an RDF response,
we can directly combine that response with the composition,
thereby instantiating a~result placeholder with the actual results.
This is the core mechanism of the execution process
detailed in \cref{sec:PrePostProof}.
Furthermore, the closeness of RDF responses to N3~rules
plays an important role in description generation,
as discussed in \cref{sec:Generation}.

\subsubsection{Arguments Related to Logic Programming}
In order to be suitable for the approach we describe in this paper, a logical framework has to fulfill certain requirements:

\begin{description}
 \item[Semantic Web compatibility] As described above it is crucial for a format describing RESTful Web APIs to act directly on the Semantic Web. Logical knowledge
 in the Semantic Web is usually expressed in RDF triples. The chosen logic should natively support this format. 
 \item[Rules] Our framework describes the possible use of hypermedia APIs using rules. We therefore need a logic which supports rules.
 \item[Proofs] Our approach makes use of proofs produced by a reasoner. These proofs are then used for further reasoning. The logic of choice should 
 be supported by a reasoner which is able to provide such proofs. Furthermore there needs to be a format to express proofs using the logic.
 \item[Conjunctions in the head of a rule] While conjunctions in the head of a rule are normally considered syntactic sugar \cite{sugar}, they become important if 
 rules appear in a proof for the following reason:
 in our approach we use the instantiated conjunction in the head of a rule as it occurs in a proof to infer concrete practical instructions. 
 To understand this idea consider a simple example: imagine that we have a source $s$ which is aware of 
 the birth date of any person. In a first order style
 we could describe that using the following rule
 \[
 \forall x:  \text{person}(x)\implies (\exists y: \text{birthDateOf}(x, y) \wedge \text{canBeAskedFor}(s,y))
 \]
If we now ask for the birth date of $\text{person}(\text{leonhard\_euler})$,
we get as a result that there exists a birth date: \[\exists y: \text{birthDateOf}(\text{leonhard\_euler}, y).\] 
But the present rule derives more, we also get an instruction how to get this particular information: 
\[
 \exists y: \text{birthDateOf}(\text{leonhard\_euler}, y) \wedge \text{canBeAskedFor}(s,y)
\]
As one rule derives both, the fact that there exists a birth date and an actual instruction how to get it, both things would appear in a proof, while
in case of having rules with atomic heads we would have to specifically ask for the instruction.
%
%
But there could be many ways 
to obtain 
that birth date: maybe we could invoke a service calculating it by knowing the exact age of the person and the date of death, 
maybe it is an even more complex process of first asking for the names of the
concrete sources and then asking these sources for information. Such things are not known beforehand.
Therefore it is not a solution to simply include the \emph{canBeAskedFor} predicate in the query.
 %
This is the same in our case where these \emph{sources} are web services. 
By using only rules with (quantified) atomic heads we would lose relevant information in the proof. 
The logic needs thus to support conjunction in the head of a rule. 
 \item[Existential quantification in the head of a rule]
The algorithm presented in this paper employs proofs to make, execute and adjust a plan of sequent API calls in order to achieve a given goal. 
The possible calls are described using rules with existentially quantified variables in their head. 
This can be understood as follows: given a certain situation, there exists a call which can be executed to obtain a new situation. While some 
details of this ``new situation'' are already clear before executing the call, some things can only be known afterwards.
This kind of uncertainty is crucial for our algorithm. The logical framework should support existentially quantified variables in the head of rules. 
\end{description}


Given the first requirement, we take a closer look into Semantic Web frameworks:
RDF, RDFS and OWL do not natively support rules, but there are several rule formats defined on top of them. 
Among them, SWRL \cite{swrl}, WRL \cite{wrl} or RIF \cite{rif}.

The Semantic Web Rule Language (SWRL) was built on top of OWL and can be understood as its rule extension. 
Rules are expressed in terms of OWL concepts
(classes, properties, individuals, etc.). SWRL thereby inherits OWL's strong separation between classes and individuals which can form a burden when
reasoning over plain RDF data.
To the best of our knowledge there is no reasoner which outputs complete proofs for derivations done applying SWRL rules. 
SWRL supports conjunctions in the head of a rule but it does not allow new existential variables in that position.
Especially the missing support for existential rules and the fact that there is no exchangeable proof produced by SWRL reasoners
made us opt against this rule language.

%

The Web Rule Language (WRL) and the Rule Interchange Format (RIF) are both based on Frame Logic (F-Logic) \cite{flogic}.
F-Logic was invented to combine object oriented and declarative programming. It extends classical predicate calculus with the concepts of objects,
classes, and types. Object oriented ideas such as inheritance are supported as well as classical rule inference. 
This richness of features made F-Logic a good choice to base the rule-based ontology language WRL on, and---as it was designed for rule exchange---even 
a better choice for RIF. Both, WRL and RIF have variants which support conjunction in the head of a rule. Neither the Basic Logic Dialect (BLD) of RIF nor
WRL do allow new existentials in the head of a rule. The direct reasoning support for both formats is very limited, 
the IRIS system \cite{iris} provides reasoner for both but does not support proofs. 
As far as we know there is also no other reasoner which produces proofs based on RIF rules.
Nevertheless, RIF can indeed serve as an exchange format for rules.


In contrast to the above mentioned formats, N3 Logic fulfills all of the requirements mentioned above. It is a rule language which easily combines with RDF, RDFS 
and---via OWL RL---also with OWL. The two most common reasoners for N3, cwm and EYE, 
both output and check proofs. N3 rules can have conjunctions and existentially quantified variables in their head. 
Therefore, we chose N3 for our purposes.


Note that by fulfilling the last requirement, support of existential variables in the head of rules, Notation3 Logic is strongly 
related to 
the Datalog$^{\pm}$ framework \cite{datalogpm,exrules}.
Several recent approaches support reasoning with existential rules on top of  ontologies as for example 
Graal \cite{graal} or IRIS$^\pm$ \cite{irispm}. These implementations are very promising, but still under development. Their rule format is
not as close to RDF as N3 is
and there is no support for proofs yet.


\section{Describing Hypermedia APIs with \restdesc}
\label{sec:RESTdesc}
\subsection{Example Use Case}\label{usecase}
As a~guiding example, we will introduce an exemplary hypermedia API
in the domain of image processing.
It offers functionality such as the following:
\begin{itemize}
  \item uploading images
  \item resizing an image
  \item changing the colors of an image
  \item combining multiple images
  \item \ldots
\end{itemize}
Since hypermedia APIs follow the REST architectural constraints,
this functionality is realized in a~resource-oriented, hypermedia-driven way.
For instance, the API does \emph{not} offer methods such as \verb!resizeImage!;
rather, it exposes \verb!image! and \verb!thumbnail! resources,
which are connected to each other through hyperlinks and/or forms.
In order to enable machines to interpret responses of this API,
these resources are available as representations in the machine-readable format RDF.

The problem statement is to make a~generic Web client
that can autonomously reach a~complex goal such as
``obtain a~scaled, black-and-white version of \verb!image.jpg!'',
without hard-coding any knowledge about this API.
The client should be able to execute complex instructions
on (combinations of) other hypermedia APIs as well.

Since this API is a~hypermedia API,
we cannot create a~detailed plan in advance,
because the exact steps are not known at design time.
However, at the same time,
an automated client cannot only follow hypermedia controls,
because there would be no way to know
whether the controls it chooses lead to the given goal.
In other words,
while hypermedia gives machines access to resources via links,
it does not explain them the functionality offered through those links.

\restdesc descriptions~\cite{verborgh_wsrest_2012,verborgh_mtap_2013}
allow to express the functionality of hypermedia APIs
by explaining the role of a~hypermedia control in an~API.
That way, if a~machine client encounters a~hypermedia control,
it can interpret the results of following it.
Furthermore, \restdesc descriptions allow the composition of a~high-level plan
that can guide machine clients through an interaction with a~hypermedia API.

In the remainder of \cref{sec:RESTdesc},
we formalize \restdesc and its underlying \nthree logic.
\Cref{sec:Proof,sec:Composition} detail hypermedia API composition and execution,
using the above image processing hypermedia API as an example.

\subsection{Formalization of Notation3}
\label{nthree}


\restdesc descriptions are expressed in the Notation3 (\nthree) rule language~\cite{N3Logic,Notation3}.
We will introduce the \nthree language and its logic,
focusing on the aspects relevant to our purposes. Our formalization is based on the formalization we gave in a previous paper \cite{semN3} and the informal
semantic descriptions given in the above mentioned sources.
\nthree augments the RDF model with symbols for quantification, implication, and statements about formulas:

\begin{definition}[Basic N3 vocabulary]
An \emph{N3 alphabet~$A$} consists of the following disjoint classes of symbols:
\begin{itemize}
\item A set $U$ of \uri symbols.
\item A set $V=V_E\mathbin{\dot{\cup}} V_U$  of (quantified) variables, with $V_E$ being the set of existential variables
and $V_U$ the set of universal variables.
\item A set $L$  of literals.
\item A boolean literal \verb!false!
\item Brackets \verb!{!, \verb!}!, \verb!(!, \verb!)!
\item Logical implication $\verb!=>!$ 
\item Period \verb!.!
\end{itemize}
\end{definition}

We define the elements of $U$ as in the corresponding specification~\cite{iri}.
\nthree allows to abbreviate URLs as prefixed names~\cite{turtle}.
Literals are strings beginning and ending with quotation marks `\verb!"!';
existentials start with `\verb!_:!', universals with `\verb!?!'.

\nthree does not distinguish between predicates and constants---%
a single \uri symbol can stand for both at the same time---%
so the first-order-concept of a~\emph{term}
has a~slightly different counterpart in \nthree: an \emph{expression}.
Since the definition of expressions (Definition~\ref{expression})
is closely related to the concept of a~formula (Definition~\ref{formula}),
the two following definitions should be considered together.

\begin{definition}[Expressions]\label{expression}
  Let $A$ be an \nthree alphabet.
  The set of \textit{expressions} $E \subset A^{*}$ is
  defined as follows:
  \begin{enumerate}
    \item Each \uri is an expression.
    \item Each variable is an expression.
    \item Each literal is an expression.
    \item \label{list} If $e_1,\ldots,e_n$ are expressions, $\verb!(!e_1 \ldots e_n\verb!)!$ is an expression. 
    \item \label{false} \verb!false! is an expression.
    \item $\verb!{ }!$ is an expression.
    \item \label{fe} If $f\in F$ is a formula, then $\verb!{!f\verb!}!$ is an expression. 
  \end{enumerate}
  The expression defined by \ref{list} is called a \textit{list}.
  We call the expressions defined by \ref{false}--\ref{fe}
  \textit{formula expressions} and denote the set of all formula expressions by $\mbox{\textit{FE}}$.
\end{definition}

Note that point \ref{fe} of the definition above makes use of formulas, which are defined as follows:

\begin{definition}[\nthree Formulas]
    \label{formula}
    The set~$F$ of \textit{\nthree formulas} over alphabet~$A$ is recursively defined as follows:
    \begin{enumerate}  
      \item \label{1} If $e_1, e_2, e_3 \in E$, then the following is a formula, called \textit{atomic} formula: \[e_1~ e_2~ e_3.\] 
      \item \label{2} If $t_1, t_2$ are formula expressions then the following is a formula, called \mbox{\textit{implication}}: \[t_{1} \verb!=>!~t_{2}.\] 
      \item \label{n} If $f_1$ and $f_2$ are formulas, then the following is a formula, called \textit{conjunction}: \[f_1 f_2\] 
    \end{enumerate}
\end{definition}

We will refer to a~formula without any variables as a~\textit{ground formula}.
Analogously, we call expressions without any variables \textit{ground expressions}.
We denote the corresponding sets by $F_g$ respectively $E_g$. 
An formula or expression which does not contain universal variables is called \emph{universal free}. 
The set of universal free formulas (possibly containing existentials) is denoted by $F_e$, the set of universal free expressions by $E_e$.

In the examples in the remainder of this paper, we will use the common \rdf shortcuts:

\begin{remark}[Syntactic variants]
\begin{itemize}
\item A formula consisting of two triple subformulas starting with the same element \verb!<d> <p> <e>. <d> <q> <f>.! can be abbreviated using a semicolon: \verb!<d> <p> <e>; <q> <f>.!\\ 
Two triple formulas sharing the first two elements  \verb!<d> <p> <e>. <d> <p> <f>.! can be abbreviated using a comma: \verb!<d> <p> <e>, <f>.!
  \item \verb![]! can be used as an expression and is a shortcut for a new existential variable. So \verb![] <p> <d>.! stands for \verb!_:x <p> <d>.!
 \item An expression of the form \verb![<p> <o>]! is a shortcut for a new existential variable \verb!_:x!,
   which is subject to the statement \verb!_:x <p> <o>.!
   So \verb! <s> <p> [<q> <o>].! stands for \verb!<s> <p> _:x. _:x <q> <o>.!
 \item \verb!a! is a~shortcut for \verb!rdf:type!~\cite{RDF}.
 \end{itemize}
\end{remark}

To emphasize the difference between brackets which form part of the \nthree vocabulary, i.e. ``\verb!(!'', ``\verb!)!'', ``\verb!{!'', and ``\verb!}!'', 
and the brackets occurring in mathematical language, we will underline the \nthree brackets in all definitions where both kinds of brackets occur.

Simple N3 triples of the form \verb!:s :p :o! can be understood as a first order formula $p(s,o)$. We call \texttt{:p} the predicate, \texttt{:s} 
 the subject and \texttt{:o} the object of a triple. More complicated constructs often contain variables:
in Notation3 existential and universal variables 
are implicitly quantified. The scope of this quantification depends on how deeply nested
a variable occurs in a formula. To be able to make statements about that we define:

\begin{definition}[Components of a formula]
Let $f\in F$ be a formula and $c: E \rightarrow 2^E$ a~function such that:
\[c(e)=\begin{cases}
  
  c(e_1)\cup\ldots\cup c(e_n) & \text{if }e=\underline{\texttt{(}}e_1 \ldots e_n\underline{\texttt{)}}\text{ is a list,}\\
  \{e\}  & \text{otherwise.}
\end{cases}\]

We define the set $\textit{comp}(f)\subset E$ of components of $f$ as follows:
 \begin{itemize}
  \item If $f$ is an atomic formula of the form $e_1~ e_2~ e_3.$, $\textit{comp}(f)=c(e_1)\cup c(e_2)\cup c(e_3)$.
  \item If $f$ is an implication of the form $t_{1} \verb!=>!~t_{2}.$, then $\textit{comp}(f)=\{t_1, t_2\}$.
  \item If $f$ is a conjunction of the form $f_1 f_2$, then $\textit{comp}(f)=\textit{comp}(f_1)\cup \textit{comp}(f_2)$.
 \end{itemize}
 Likewise, for $n\in \mathbb{N}_{>0}$, we define the components of level $n$ as:
 \begin{flalign*} 
  \textit{comp}^n(f):= &  
  \{e\in E|\exists f_1,\ldots, f_{n-1}\in F: 
   e\in \textit{comp}(f_1)\wedge  \underline{\texttt{\{}}f_1\underline{\texttt{\}}}\in \textit{comp}(f_2)\wedge \ldots\\& \wedge  \underline{\texttt{\{}}f_{n-1}\underline{\texttt{\}}}\in 
  \textit{comp}(f)\} 
\end{flalign*} 
\end{definition}

Now, we can distinguish between direct components and 
nested components. As an example take the following \nthree formula:
\begin{Verbatim}[fontsize=\normalsize] 
		:John :says {:Kurt :knows :Albert.}.
\end{Verbatim}
Direct components are \verb!:John!, \verb!:says! and \verb!{:Kurt :knows :Albert.}! while \verb!:Kurt!,  \verb!:knows! and  \verb!:Albert! are nested components of
level two. \nthree allows more complex structures, formulas like
\[
\verb! {{?x :p :a.} => {?x :q :b.}.} => {{?x :r :c.} => {?x :s :d.}.}.! 
\]
where for example the predicate \verb!:p! occurs as a component of level three are valid in \nthree. 
Such deeply nested structures require a careful treatment of scoping for variables occurring in them. 
Note for example that the above formula should be interpreted as
\[
(\forall x_1: p(x_1,a)\rightarrow q(x_1,b))\rightarrow (\forall x_2: r(x_2,c) \rightarrow s(x_2,d))
\]
and \emph{not} as
\[\forall x: ((p(x,a)\rightarrow q(x,b))\rightarrow ( r(x,c) \rightarrow s(x,d)))\]
Due to this particularities and because deeply nested structures are no requirement for our framework, we limit the
considerations of
this paper to \emph{simple formulas}
and refer the reader interested in more details to the corresponding publication \cite{semN3}.

\begin{definition}[Simple formulas]
We call an \nthree formula $f$ \emph{simple} iff for all $n \in \mathbb{N}$, $n>2$: $\operatorname{comp}^n(f)=\emptyset$
\end{definition}

Universal variables in \emph{simple formulas} can be understood as universally quantified on the top level of the formula.
The formula 
\[
 \verb!{?x :p :o1.} => {?x :q :o2.}.!
\]
is interpreted as 
\[
 \forall x: p(x, o_1)\rightarrow q(x, o_2)
\]
The scope of an existential variable is always the formula expression it occurs in as a direct component. The formula
\[
 \verb!_:x :says {_:x :knows :Albert.}.!
\]

is interpreted as 
\[
 \exists x_1: \text{says}(x_1, (\exists x_2: \text{knows}(x_2, Albert)))
\]

As the existential quantification of blank nodes, in contrast to universal quantification, 
only counts for the direct formula they occur in and not for their subordinated formulas, 
we define two ways to apply a substitution:

\begin{definition}[Substitution]
Let $A$ be an \nthree alphabet 
and $f\in F$ an \nthree formula over~$A$. 
\begin{itemize}
 \item A \emph{substitution} is a finite set of pairs of expressions $\{v_1/e_1, \ldots, v_n/e_n\}$ where each $e_i$ is an expression and each $v_i$ 
 a variable such that $v_i\neq e_i$ and 
 $v_i \neq v_j$,
 if $i\neq j$.  
 \item 
 For a formula $f$ and a substitution $\sigma=\{v_1/e_1, \ldots, v_n/e_n\}$, we obtain the \emph{component application} 
 of $\sigma$ to $f$, $f\sigma^c$, by simultaneously replacing each $v_i$ 
 which occurs as a \emph{direct component} in $f$ by the corresponding expression $e_i$. 
 \item 
 For a formula $f$ and a substitution $\sigma=\{v_1/e_1, \ldots, v_n/e_n\}$, we obtain the \emph{total application} of $\sigma$ to $f$, $f\sigma^t$, 
 by simultaneously replacing each $v_i$ 
 which occurs as a \emph{direct or nested component} in $f$ by the corresponding expression~$e_i$. 
 \end{itemize}
\end{definition}

As the definition states, component application of a substitution only changes the direct components of a formula. 
For a substitution $\mu=\{\verb!_:x!/ \verb!:Kurt!\}$ we obtain:
\begin{multline}
(\texttt{\_:x :says \{\_:x :knows :Albert.\}.})\mu^c  =\nonumber \\ (\texttt{ :Kurt :says \{\_:x :knows :Albert.\}.})\nonumber
\end{multline}
A total application of $\sigma=\{\verb!_:x!/ \verb!:Kurt!\}$ in contrast, replaces each \emph{occurrence} of a variable in a formula: 
\begin{multline}
(\texttt{?x :says \{?x :knows :Albert.\}.})\sigma^t =\\ (\texttt{:Kurt :says \{:Kurt :knows :Albert.\}.})\nonumber
\end{multline}

For simple formulas we can now define interpretation and semantics:

\begin{definition}[Interpretation]
An interpretation $\mathfrak{I}$ of
an alphabet $A$ consists of:
\begin{enumerate}
\item A set $\mathcal{D}$ called the domain of $\mathfrak{I}$.
\item A function $\mathfrak{a}: 
E\setminus V \rightarrow \mathcal{D}$ called the object function.
\item A function $\mathfrak{p}:
\mathcal{D} \rightarrow 2^{\mathcal{D} \times \mathcal{D}}$ called the predicate function.
\end{enumerate}
\end{definition}

Note that in contrast to the classical definition of \rdf-semantics \cite{RDFSemantics} our domain does not distinguish between properties (IP) 
and resources (IR). 
The definitions are nevertheless compatible, as we assume $\mathfrak{p}(p)=\emptyset\in 2^{\mathcal{D} \times \mathcal{D}}$
for all resources $p$ which are not properties (i.e. $p \in \text{IR}\setminus \text{IP}$ in the \rdf-sense). 
By extending given \rdf ground interpretation functions to Notation3 interpretation functions, 
the meaning of all valid \rdf triples can be kept in Notation3 Logic. 
%

\begin{definition}[Semantics of \nthree\label{sem_n3}]
Let $\mathfrak{I}=(\mathcal{D},\mathfrak{a,p})$ be an interpretation of $A$
and let  $f$ be a \emph{simple} formula over $A$. 
Then the following holds:
\begin{enumerate}
 \item\label{quant1} If $f$ contains universal variables, $\mathfrak{I}\models f$ iff $\mathfrak{I}\models f\sigma^t$ 
 for every substitution $\sigma: V_U\rightarrow E_e$. 
 \item \label{quant2} If $f$ is universal free and $W=\text{comp}(f)\cap V_E \neq \emptyset$, $\mathfrak{I}\models f$ 
 iff $\mathfrak{I}\models f\mu^c$ for some substitution $\mu: W\rightarrow E_g$.
  \item If $f$ is universal free and $\text{comp}(f)\cap V_E=\emptyset$:
  \begin{enumerate}
   \item If $f$ is an atomic formula $c_1\, p\, c_2$, then  $\mathfrak{I} \models c_1\, p\, c_2$. 
  iff $(\mathfrak{a}(c_1),\mathfrak{a}(c_2))\in\mathfrak{p}(\mathfrak{a}(p))$.
  \item If $f$ is a conjunction $f_1f_2$, then  $\mathfrak{I}\models f_1 f_2$ iff $\mathfrak{I}\models f_1$ and $\mathfrak{I}\models f_2$.\label{conj}
  \item If $f$ is an implication\label{implication} 
  \begin{itemize}
  \item $\mathfrak{I} \models \verb!{! f_1 \verb!}! \verb!=>! \verb!{! f_2 \verb!}!$ iff $\mathfrak{I} \models f_2$ if $\mathfrak{I} \models f_1$.
  \item \label{fal2} 
  $\mathfrak{I} \models \verb!{! f_1 \verb!}! \,\verb!=>!\, \verb!false!$. iff $\mathfrak{I} \not\models f_1$.
   \item $\mathfrak{I} \models \verb!{ }! \verb!=>! \verb!{! f_2 \verb!}!$. iff $\mathfrak{I} \models f_2$.
   \item $\mathfrak{I} \models \texttt{false => }e$ and $\mathfrak{I}\models e\texttt{ => \{~\}}.$, for all $e\in E$, 
   \end{itemize}
 \end{enumerate}
\end{enumerate}
  
\end{definition}

Note that by first handling universal variables (point \ref{quant1} in the definition) and then treating existentials (point \ref{quant2}) the definition makes sure 
that in case of conflicts the universal quantifier is outside of the existential. In N3 the statement
\[
 \verb!?x :loves _:y.!
\]
has to be interpreted as $\forall x \exists y: \text{loves}(x,y)
\text{ and \emph{not} as }
\exists y \forall x : \text{loves}(x,y)$.

\noindent
We now define a model:

\begin{definition}[Model]
Let $\Phi$ be a set of \nthree formulas. We call an interpretation $\mathfrak{I}=(\mathcal{D},\mathfrak{a,p})$ a \textit{model} of $\Phi$ iff $\mathfrak{I}\models f$ for every formula $f\in \Phi$.
\end{definition}
As in first order logic, we can define the notion of logical implication:

\begin{definition}[Logical implication]\label{log_impl}
Let $\Phi$ be a set of \nthree formulas  and $\phi$ a formula over the same \nthree alphabet $A$. We say that $\Phi$ (logical) implies 
$\phi$ ($\Phi \models \phi$) iff every
model $\mathfrak{I}\models \Phi$ is also a model of $\phi$.
\end{definition}




\subsection{\restdesc Descriptions}\label{rd}

\restdesc descriptions are designed to explain how a hypermedia API can be used to perform a specific action. 
Such a process of using an API consists of different steps:
given all needed information, the client sends an \http request to the hypermedia API on a server.
The API interprets the request and,
if possible, reacts by fulfilling the indicated task
or retrieving the information requested. The server sends a~response and thereby creates a new situation. As mentioned in \cref{api}, this process 
is called an API operation.
Hypermedia APIs are commonly able to perform different kinds of operations.

To describe such an operation in \nthree
a formalization of \http
is needed. 
We use the \rdf vocabulary as defined in the corresponding W3C Working Draft \cite{httprdf}.
In order to facilitate the understanding of the following sections, we give a short overview of the \http predicates used in this paper.

\begin{definition}[\http predicates\label{httppr}]
Let $A$ be an \nthree alphabet,
$M$ be a set of \http method names,
$\{ \texttt{"GET"},\texttt{"POST"},\texttt{"PUT"},\linebreak[1]
\texttt{"DELETE"},\texttt{"HEAD"},\texttt{"PATCH"}\}\subset M \subset L$,
$\verb!u!\in U$ an \http message,
$\verb!v!\in U\cup L$
and $\mathfrak{I}=(\mathcal{D},\mathfrak{a,p})$ an interpretation of $A$.
Then:

\begin{itemize}
 \item  $\mathfrak{I}\models$ \verb!u http:headers v.! iff \texttt{v} is an \http header of \verb!u!.
 \item $\mathfrak{I}\models$ \verb!u http:body v.! iff \verb!v! is the \http body of \verb!u!. 
  \item $\mathfrak{I}\models$ \verb!u http:methodName v.! iff $\verb!u!$ is a request and 
  $\verb!v!\in M$ its method name.
  \item $\mathfrak{I}\models$ \verb!u http:requestURI v.! iff $\verb!u!$ is a request and $\verb!v!$ is its \URL.
  \item $\mathfrak{I}\models$ \verb!u http:resp v.! iff $\verb!u!$ is a request and $\verb!v!$ is its responding \http message.
\end{itemize}
\end{definition}

%

This vocabulary can be used to describe \http-requests.  Such a request must always have a method name and a request \uri.
Using owl-vocabulary, such requests can be defined by the following class:
\begin{verbatim}
http:Request rdfs:subClassOf http:Message, 
    [
    rdf:type     owl:Restriction ;
    owl:onProperty http:methodName ;
    owl:qualifiedCardinality "1"^^xsd:nonNegativeInteger
    ],
    [ 
    rdf:type owl:Restriction ;
    owl:onProperty http:requestURI ;
    owl:qualifiedCardinality "1"^^xsd:nonNegativeInteger
    ] .
\end{verbatim}
These two properties also have to be specified in an \http request description:

\begin{definition}[\http request description]
  \label{def:HttpRequestDescription}
 Let $A$ be an \nthree alphabet which contains a set $H$ of \http predicates including those defined in Definition~\ref{httppr}. 
  \begin{itemize}
\item An \emph{\http request description} is a conjunction $f=f_1 f_2 \ldots f_n\in F$ of atomic formulas with the following properties:
\begin{itemize}
 \item All atomic formulas $f_i$ share the same existential variable $\verb!_:x!\in V_E$ as a subject.
 \item The predicate of each formula $f_i$ is an \http-predicate, i.e. $f_i$ has the form $\verb!_:x :h!_i \verb!:o!_i\verb!.!$ with $\verb!:h!_i\in H$.
 \item The conjunction $f$ contains one atomic formula $f_i$ with the predicate\linebreak \verb! http:methodName!
 and one formula $f_j$ with the predicate\linebreak 
 \verb! http:requestURI!.
 \item The object of every atomic formula $f_i=\verb!_:x :h!_i \verb!:o!_i\verb!.!$ with $\verb!h!_i \neq \verb!http:resp!$ is either a universal variable, 
 a URI or a literal, $\verb!:o!_i\in V_U\cup U\cup L$.
\end{itemize}

%
\item  We call an \http request description $f=f_1\ldots f_n$ \textit{sufficiently specified},
if  $\texttt{:o}_i\in E_g$ for every triple $f_i=\texttt{\_:x :h}_i \texttt{:h}_i$ with $\texttt{:p}_i\neq \texttt{http:resp}$. 
\end{itemize}

\end{definition}

The definition reflects 
the syntactical requirements
to an \http request description, 
it should contain  the URL and the method name of the described request and it can contain additional information which can be
described using the \http-predicates.
If the object of these formulas are instantiated, i.e. sufficiently specified, they can be sent to a server and, 
if they contain all necessary information, executed by an API which will return the \http response.
\restdesc descriptions enable us to specify
the intended functionality of a~hypermedia API's operations:

\begin{definition}[\restdesc description]
Let $A$ be an \nthree alphabet containing the predicates defined in Definition \ref{httppr}, $F$ the set of formulas over $A$.
A \restdesc description $f\in F$ of a hypermedia API operation is a \emph{simple} \nthree formula 
of the form:
\[	  \verb!{ precondition } => { http-request  postcondition }.!\]
where \verb!precondition!, \verb!http-request! and \verb!postcondition! are \nthree formulas over $A$ with the following properties:
\begin{enumerate}
 \item \emph{\texttt{precondition}} describes the resources needed to execute the operation and does not contain 
 any existential variable. 
 \item \emph{\texttt{http-request}} is an \http request description which describes a request which can be used to obtain the desired 
 result of executing the operation. 
 It contains no triple having the same subject as the \verb!http-request!.
All universal variables which occur in \verb!http-request!
 do also occur in \verb!precondition!. 
 
 \item \emph{\texttt{postcondition}} describes one or more results obtained by the execution of the operation. All universal variables contained in \verb!postcondition! 
also occur in \verb!precondition!. 
\end{enumerate}

\end{definition}
By making sure that the subject of any triple in the postcondition is different than the subject of the request, we make both syntactical distinguishable. 
Note that a RESTdesc description is an existential rule as defined in the Datalog$^{\pm}$ framework \cite{datalogpm}:
our restriction on universal variables,
that all universals in the head of the rule should also occur in the body,
is very similar to Datalog \cite{datalog}
and our rules 
 allow (and expect) new existentials in the consequence. 

The reasons for the restrictions will become more clear in \cref{sec:Reasoning},
where we show that for every ground instance of the precondition,
the \http request is sufficiently specified and the postcondition will not contain any universal variables.
From an operational point of view, the remaining existential variables in the postcondition
are those which are expected to be grounded through the execution of the \http request.

\begin{lstlisting}[
  float=t,
  caption={\restdesc description of the action ``obtaining a~thumbnail'' \emph{(desc\_thumbnail.n3)}},
  label=lst:Thumbnail]
§\textcolor{gray}{@prefix dbpedia: <http://dbpedia.org/resource/>.}§
§\textcolor{gray}{@prefix dbpedia-owl: <http://dbpedia.org/ontology/>.}§
§\textcolor{gray}{@prefix ex: <http://example.org/image\#>.}§
§\textcolor{gray}{@prefix http: <http://www.w3.org/2011/http\#>.}§

{ ?image ex:smallThumbnail ?thumbnail. }
=>
{
  _:request http:methodName "GET";
            http:requestURI ?thumbnail;
            http:resp [ http:body ?thumbnail ].

  ?image dbpedia-owl:thumbnail ?thumbnail.
  ?thumbnail a dbpedia:Image;
             dbpedia-owl:height 80.0.
}.
\end{lstlisting}

\Cref{lst:Thumbnail} shows a~description
that explains the \verb!smallThumbnail! relation in a~hypermedia API.
The \emph{precondition} demands the existence of
a~\verb!smallThumbnail! hyperlink
between an \verb!?image! resource
and a~\verb!?thumbnail! resource.
The \http request
is a~\verb!GET! request to the \URL of \verb!?thumbnail!.
The response to this request will be a~representation of \verb!?thumbnail!.
These characteristics of this representation are detailed in the remainder of the \emph{postcondition}.
This states that the original \verb!?image!
will be in a~\verb!thumbnail! relationship
(the meaning of which is defined by the \dbpedia ontology~\citeS{DBpedia})
with \verb!?thumbnail!.
Furthermore, \verb!?thumbnail! will be an \verb!Image!
and have a~\verb!height! of \verb!80.0!.

There are two different ways to interpret this description:
First the declarative, static way as defined in \cref{nthree},
which could be phrased as
\enquote{the existence of the \texttt{smallThumbnail} relationship
implies the existence of a~\texttt{GET} request
which leads to an~80px-high thumbnail of this image.}

The second interpretation is the operational, dynamic way.
In this case, a~software agent has a~description of the world,
against which the description is \emph{instantiated},
i.e., the rule is applied.

Thus, given a~concrete set of triples, such as:
\begin{Verbatim}
    </photos/37> ex:smallThumbnail </photos/37/thumb>.
\end{Verbatim}
the description in \cref{lst:Thumbnail} would be instantiated to: 
\begin{Verbatim}
    _:request http:methodName "GET";
              http:requestURI </photos/37/thumb>;
              http:resp [ http:body </photos/37/thumb> ].

    </photos/37> dbpedia-owl:thumbnail </photos/37/thumb>.
    </photos/37/thumb> a dbpedia:Image;
                       dbpedia-owl:height 80.0.
\end{Verbatim}
Thereby, the description has been instantiated into a~concrete \http request
that can be executed by the agent.
Note how this instantiation directly results in RDF triples,
which can be interpreted by any RDF-compatible client.
The request has been sufficiently specified
as defined in Definition~\ref{def:HttpRequestDescription}.
In addition, the instantiated postcondition
explains the properties realized by this concrete request.
Here, an \http \verb!GET! request to \url{/photos/37/thumb}
will result in a~thumbnail of the image \url{/photos/37}
that will have a~height of 80~pixels.
This dynamic interpretation is helpful to agents
that want to understand the impact
of performing a~certain action on resources they have at their disposition.

\begin{lstlisting}[
  float=b,
  caption={\restdesc description of the action ``uploading an image'' \emph{(desc\_images.n3)}},
  label=lst:Images]
§\textcolor{gray}{@prefix dbpedia: <http://dbpedia.org/resource/>.}§
§\textcolor{gray}{@prefix dbpedia-owl: <http://dbpedia.org/ontology/>.}§
§\textcolor{gray}{@prefix ex: <http://example.org/image\#>.}§
§\textcolor{gray}{@prefix http: <http://www.w3.org/2011/http\#>.}§

{ ?image a dbpedia:Image. }
=>
{
  _:request http:methodName "POST";
            http:requestURI "/images/";
            http:body ?image;
            http:resp [ http:body ?image ].
  ?image ex:comments _:comments;
         ex:smallThumbnail _:thumb.
}
\end{lstlisting}

\restdesc descriptions are not limited to \verb!GET! requests.
They can also describe state-changing operations,
for instance, those realized through the \verb!POST! method.
\Cref{lst:Images} shows a~description for an image upload action.
The postconditions contain existential variables that are not referenced
(\verb!_:comments! and \verb!_:thumb!),
which might appear strange at first sight.
However, these triples are important to an agent
as they convey an \emph{expectation} of what happens when an image is uploaded.
Concretely, any uploaded image will receive a~\verb!comments! link
and a~\verb!smallThumbnail! link.
Even though the exact values will only be known at runtime
when the actual \verb!POST! request is executed,
at design-time, we are able to determine that there will be several links.
The meaning of those links is in turn expressed by other descriptions,
such as the one in \cref{lst:Thumbnail} discussed~above.

\section{Proof of a~Composition's Correctness}
\label{sec:Proof}

%
%

\subsection{Compositions of Hypermedia APIs}
\label{subsec:Definition}
Having introduced a formal way to describe the function of hypermedia APIs in the last chapter,
we will now focus on the proofs which can be created using these rules.
More concretely, we examine proofs which confirm that a certain combination of API calls brings us to a desired goal.
We therefore take a closer look at the problem itself.
Given a set of possible API operations,
we want to achieve a~goal from an initial state.
Furthermore, we might have some additional knowledge that can be incorporated.
The above can be expressed in \nthree as follows:

\begin{definition}[API composition problem]\label{apicomp}
Let~$F$ be the set of simple \nthree formulas over an alphabet~$A$ which contains the predicates defined in definition \ref{httppr}.
An \textit{API composition problem} consists of the following formulas:
\begin{itemize}
 \item A set $H\subset F_g$ of ground formulas capturing all resource and application states the client is currently aware of,  the \textit{initial state}.
 \item 
 A formula $g\in F$ with $\text{comp}^2(g)=\emptyset$, which does not contain existential variables, the \textit{goal state} which
  indicates on a symbolic level what the client wants to achieve.
 \item A set $R$ of
 \restdesc descriptions or conjunctions of  \restdesc descriptions, describing
 all hypermedia APIs available to the client, the \textit{description formulas}.
 
 \item 
 A (possibly empty) set of \nthree formulas $B$, the \textit{background knowledge}, 
 where each $b\in B$ is either a ground formula or an implication $e_1\verb!=>! e_2.$ 
 which does not contain existential variables and 
 where each universal variable $e_2$ contains 
does also occur in $e_1$.
\end{itemize}
\end{definition}

Note that we put syntactical restrictions on our definitions: as already mentioned in \Cref{rd}, RESTdesc descriptions are
existential rules. The constraints put on the background knowledge make the rules contained in it expressible in Datalog \cite{datalog}. The initial state
contains only ground formulas and as the goal does not contain nested constructions or existential variables it can also be expressed in a Datalog rule.
This makes the whole problem at our disposal expressible in Datalog$^\pm$ \cite{datalogpm}. The reason for this restrictions 
will become clear in \cref{sec:Reasoning},
\\
If we talk about a proof as explained above, we mean evidence for the fact that
from $H\cup R \cup B$ follows $g'$, where $g'$ is a valid instance of $g$. 
As a normal hypermedia API composition problem tries to actually achieve real states
and obtain instantiated objects, our final target is to make $g'$ ground.

\subsection{Pre-proofs versus Post-proofs}
\label{sec:PrePostProof}
The focus of this section is not how to create proofs
but how to verify their correctness,
given that creation has already been performed.
This is not unlike the notion of proof in the classical Semantic Web vision~\cite{SemanticWeb},
where it is defined as~a means to assert the validity of a~piece of (static) information.
In this article, we extend this classical notion or proofs
to also include \emph{dynamic} information,
{\it i.e.,} data generated by Web~APIs.
As a~consequence of this dynamic nature,
we introduce two different kinds of proofs for an API composition problem $(H, g, R, B)$ as defined in Definition \ref{apicomp}:


\begin{description}
\item [a~{\bf pre-execution proof} \emph{(``pre-proof'')}]\hspace{-1ex},
      in which the assumption is made that execution of all API operations
      will behave as~expected, i.e., a proof in a classical sense which provides evidence for
      \[ H\cup R \cup B \models g'\]
\item [a~{\bf post-execution proof} \emph{(``post-proof'')}]\hspace{-1ex},
      in which an additional evidence for the goal 
      is provided by
      the API operations' actual execution results,
      which are purely static data. 
      This means 
      the resulting proof itself confirms that 
      \[H \cup R \cup B \cup \{\text{execution results}\} \models g'' \]
\end{description}
with $g'$ and $g''$ being instances of $g$.
Note that technically speaking,
for the same API operation every pre-proof is also a post-proof but,
if its execution actually yields a non-empty result,
not every post-proof is a pre-proof.
We are especially interested in the post-proofs of an API operation which are not its pre-proofs,
and call those proofs
\textbf{proper post-execution proofs}.
In other words, proper post proofs are those proofs
that actually make use of the information gained by the API call.
Intuitively, we expect those proofs to be shorter than the initial pre-proof,
as they have more relevant knowledge at their disposal.

The distinction between pre- and  (proper) post-proof exists because,
although error handling is possible,
one can never \emph{guarantee} that a~composition that has proven to work in theory
will \emph{always} and reliably achieve the desired result in practice,
since the individual steps can fail.
Some errors (such as disk failures or power outages) cannot be predicted
and may cause a~composition not to reach a~goal that would normally be possible.
Furthermore, a~composition can only be as adequate
as its individual descriptions, which could contain mistakes.
Therefore, the pre-proof necessarily has to make the additional assumption
that all APIs will function according to their description.
The pre-proof's objective thus becomes:
\emph{``assuming correct behavior of all APIs,
      the composition must lead to the fulfillment of the~goal.''}

While a~pre-proof can be validated before a~composition's execution,
the creation and validation of a~proper post-proof can only happen
when the execution's results are available.
At that stage, however, the environment's nature is no longer dynamic,
since the APIs' results are effectively available as data.
A~proper post-proof is therefore equivalent to a~data-based proof,
wherein the executed API operations contribute to the provenance information.
This provenance can be used to link the proper post-proof to the pre-proof,
indicating whether the non-failure assumption has corresponded to reality.
In the ideal case, this assumption indeed holds
and the proper post-proof is essentially a~revision of the pre-proof
in which the actual values returned by the hypermedia APIs are filled in.
The objective of the post-proof is thus
\emph{``given the execution results of some API operations,
      the composition must lead to the fulfillment of the~goal.''}
Thereby, a~proper post-proof after one operation
becomes the pre-proof of the next operation,
as indicated in \cref{fig:Proofs}.

Regular proofs do not contain dynamic information that needs to be obtained at runtime.
The extension to pre-proofs that contain dynamic information,
necessary to verify the correctness of a~composition \emph{before} it is executed,
requires a~mechanism to express when API operations are performed.
\restdesc descriptions can be considered rules
that \emph{simulate} the execution of a~hypermedia API,
using existentially quantified variables as placeholders for the API's results,
which are still unknown at the time the pre-proof is to be verified.

\begin{figure}[t]
  \begin{tikzpicture}[
      line width=.8pt,
      label distance = -1pt,
      state/.style = {
        circle, fill=white, draw=black,
        prefix after command={\pgfextra{\tikzset{
          every label/.style={black}, font=\small,
        }}},
      },
      request/.style = {
        rounded rectangle, draw=black, fill=white,
        text width=60pt, text depth=12pt, inner sep=5pt, align=center,
        font=\ttfamily\fontsize{9pt}{11pt}\selectfont\bfseries,
      },
      connector/.style = {
        black, line width=.8pt, -angle 60,
      },
      proof/.style = {
        black, draw,
        shorten <=2pt, shorten >=11pt, latex-,
        font=\small, inner sep=-2pt,
      },
      pretopost/.style = {
        black, draw, line width=.8pt, densely dotted, ->,
      },
      preproof/.style  = { align=left,  anchor=west, xshift=-.9pt },
      postproof/.style = { align=right, anchor=east, xshift=-.4pt },
      node distance = 90pt,
    ]
    \node[state,label=below:initial state] (Initial) {};
    \node[request,right=30pt of Initial] (POST) {POST\\ /images/};
    \node[request,right=of POST] (GET) {GET\\ /2748/thumb/};
    \node[state,label=below:goal,right=30pt of GET] (Goal) {};
    \begin{pgfonlayer}{background}
      \draw[connector] (Initial) -- (Goal);
    \end{pgfonlayer}
    \path[proof] (POST.west) ++(-15pt,0pt) -- ++(0pt,28pt)
      node[preproof]  {\textbf{\itshape pre-proof~1}\\[-2pt]\scriptsize API calls:~\kern-.2ex 2};
    \path[proof] (POST.east) ++(+20pt,0pt) -- ++(0pt,28pt)
      node[postproof] {\textbf{\itshape post-proof~1}\\[-2pt]\scriptsize API calls:~\kern-.2ex 1};
    \path[proof] (GET.west)  ++(-20pt,0pt) -- ++(0pt,28pt)
      node[preproof]  {\textbf{\itshape pre-proof~2}\\[-2pt]\scriptsize API calls:~\kern-.2ex 1};
    \path[proof] (GET.east)  ++(15pt,0pt)  -- ++(0pt,28pt)
      node[postproof, xshift=1.4pt] {\textbf{\itshape post-proof~2}\\[-2pt]\scriptsize API calls:~\kern-.2ex 0};
    \makeatletter
    \pgfsys@transformcm{1}{0}{0.23}{1}{-8pt}{0pt}
      \path[pretopost] ($(POST.east) + (26.0pt,32.2pt)$) -- ($(GET.west) + (-26pt,32.2pt)$);
    \makeatother
  \end{tikzpicture}
  \caption{A~proper post-proof of one operation becomes the pre-proof of the next.}
  \label{fig:Proofs}
\end{figure}
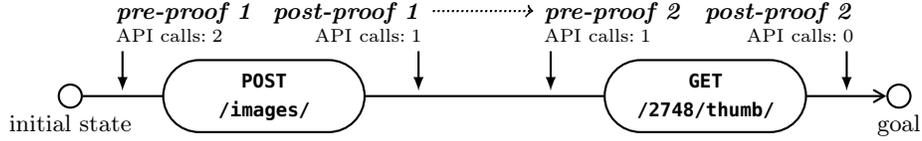

\subsection{Anatomy of a~Hypermedia API Composition Proof}
\label{sec:ProofAnatomy}
The \nthree proof vocabulary created in the context of the Semantic Web Application Platform (SWAP) \cite{SWAP} enables us to formalize proofs in a~machine-readable way.
This subsection gives a short introduction into the terminology used and the resulting proofs,
focusing on the aspects relevant to our purposes.

A proof is a conjunction of \nthree formulas describing 
inference steps
a reasoner has performed to come to a certain conclusion, so called \textit{proof steps}.

The vocabulary distinguishes between four different kinds of proof steps. We write them as deduction rules, using ``$\vdash$''.

\begin{definition}[Proof steps]\label{proofsteps}
Let $F$ be the set of simple formulas over an \nthree alphabet $A$, $\Gamma\subset F$ a set of formulas and
$f,f_1,f_2, g\in F$. 
A \textit{proof step} is one of the following inference rules:
\begin{enumerate}
 \item \emph{Axiom:} If $f \in \Gamma$ then $\Gamma \vdash f$.
 \item \emph{Conjunction elimination:} If $\Gamma \vdash f_1f_2$ then $\Gamma \vdash f_1$ and $\Gamma \vdash f_2$.
 \item \emph{Conjunction introduction:} Let $\Gamma\vdash f_1$ and $\Gamma \vdash f_2$ and 
 let \[\sigma: V_U\rightarrow V_U\setminus(\text{comp}^1(f_1)\cup\text{comp}^2(f_1)) 
 \text{ and } 
 \mu:V_E\rightarrow V_E\setminus \text{comp}(f_1)\] be substitutions. Let $f_2'= f_2\sigma^t\mu^c$ 
 %
 %
then \[\Gamma \vdash f_1f_2'\]
 \item \emph{Generalized modus ponens:} If $\Gamma \vdash \verb!{!f_1\verb!}=>{!f_2\verb!}.!$ and $\Gamma \vdash g$ 
 and there exists a substitution $\sigma:\text{comp}(f_1)\cap V_U \rightarrow E_e$ such that 
 \[
  (\verb!{!f_1\verb!}=>{!f_2\verb!}.!)\sigma^t= (\verb!{!f_1'\verb!}=>{!f_2'\verb!}.!)\text{ and }
 f'_1=g\] 
then $\Gamma \vdash f_2'$.
 \end{enumerate}
\end{definition}

\begin{theorem}[Correctness of proof calculus]\label{correctness}
Let $\Phi$ be a set of \nthree formulas  and $\phi$ a formula over the same \nthree alphabet $A$.  Then the following holds:
\[ \text{If } \Phi \vdash \phi \text{ then } \Phi \models \phi.\]
\end{theorem}


\begin{proof}

We prove that every proof step is correct.
\begin{enumerate}
\leftskip2.5em 
 \item \emph{Axiom:} For the axiom step the claim is trivial, as it corresponds to Definition~\ref{log_impl}.
 \item \emph{Conjunction elimination:}
 Let $\Phi\models f_1 f_2 $ and let  $\mathfrak{I}=(\mathcal{D},\mathfrak{a,p})$ be a model for $\Phi$ and $f_1 f_2$. 
 If $f_1 f_2$ is universal free and $\text{comp}(f_1f_2)\cap V_E=\emptyset$, the claim follows immediately 
 from Definition  \ref{sem_n3}.\ref{conj}. \\
 If $f_1f_2$ universal free and $\text{comp}(f_1f_2)\neq \emptyset$ let $\mu:\text{comp}(f_1f_2)\rightarrow E_g$ be a substitution such that
 $\mathfrak{I}\models (f_1f_2)\mu^c$. Then $\mathfrak{I}\models (f_1\mu^c) (f_2\mu^c)$, thus  $\mathfrak{I}\models (f_1\mu^c)$ and 
 $\mathfrak{I}\models (f_2\mu^c)$.\\
 If $f_1f_2$ are not universal free, then $\mathfrak{I}\models (f_1f_2)\sigma^t$ for all substitutions $\sigma:V_U\rightarrow E_e$. 
 The claim follows by the same 
 argument as above.
 \item \emph{Conjunction introduction:} 
 Let $\Phi\models f_1$, $\Phi\models f_2$ and let  $\mathfrak{I}=(\mathcal{D},\mathfrak{a,p})$ be a model for $\Phi$, $f_1$ and $f_2$. 
As the renaming substitutions $\sigma$ and $\mu$ do not change the meaning of a formula, for  $f_2'= f_2\sigma^t\mu^c$ the following holds:
$\mathfrak{I}\models f_2'$. 
It immediately follows that $\mathfrak{I}\models f_1f_2'$. 
\item \emph{Generalized modus ponens:} the claim follows directly from Definitions \ref{sem_n3}.\ref{quant1} 
and~\ref{sem_n3}.\ref{fal2}.
\end{enumerate}
\end{proof}

Applied on a API composition problem, we get the following consequence:
\begin{corollary}[Correctness of API composition proofs]
Let $(H,g,R,B)$ be an API composition problem and $g'$ an instance  of $g$ then the following holds:
\[\text{If }H\cup R \cup B \vdash g' \text{ then }H\cup R \cup B \models g'\]
\end{corollary}


%

We will examine the generalized modus ponens in more detail,
as this is the proof step where implication rules, 
such as \restdesc descriptions, are applied.

\label{sec:Reasoning}
\begin{lemma}\label{lemma:Reasoning}
Let $A$ be an \nthree alphabet, 
$f\in F_g$ a simple ground formula and $\verb!{!f_1\verb!}=>{!f_2\verb!}!\in F$ a simple implication formula 
where all 
universal variables which occur in $f_2$ 
also occur in $f_1$. If the generalized modus ponens is applicable to $f$ and $\verb!{!f_1\verb!}=>{!f_2\verb!}!$ then
the resulting formula does not contain universal variables.
\end{lemma}

\begin{proof}
 Let $\sigma:V_U\rightarrow E_e$ be a substitution 
 such that  $(\verb!{!f_1\verb!}=>{!f_2\verb!}.!)\sigma^t= \verb!{!f\verb!}=>{!f_2'\verb!}.!$ 
As $f$ is a ground formula, $\text{range}(\sigma|_{V_U\cap(\text{comp}^1(f_1)\cup \text{comp}^2(f_1))}) \subset E_g$. 
Due to the condition that every universal variable of $f_2$ 
is also in $f_1$, i.e. 
\[((\text{comp}^1(f_2)\cup \text{comp}^2(f_2))\cap V_U )\subset (\text{comp}^1(f_1)\cup \text{comp}^2(f_1))\cap V_U)\]
the claim follows. 
%
\end{proof}

As \http requests in \restdesc descriptions only contain
one leading existential to represent the \http message, and 
\restdesc descriptions 
fulfill the conditions of Lemma~\ref{lemma:Reasoning}
we arrive at the following consequence:

\begin{corollary}
\label{corollary}
Every application of a ~restdesc description to a ground formula results in a~sufficiently specified \http request 
and a postcondition which does not contain any universal variables.
\end{corollary}

The first step of Definition \ref{proofsteps} includes from a technical point of view also the parsing of a source.
In the \nthree proof vocabulary%
\footnote{The vocabulary's RDF definition can be found at: \url{http://www.w3.org/2000/10/swap/reason\#}.}
we will discuss next,
this step is therefore named after this action.

\begin{definition}[Proof vocabulary]
Let $A$ be an \nthree alphabet and $\mathfrak{I}=(\mathcal{D}, \mathfrak{a}, \mathfrak{p})$ be an 
interpretation of its formulas. Let $\verb!x!, \verb!y!,  \verb!y!_1, \ldots, \verb!y!_n \in U$ be \nthree representations of 
proof steps and $\verb!z!_1,\verb!z!_2,\verb!z!_3\in U$.
\begin{enumerate}
 \item 
 Proof step types:
 \begin{itemize}
\item $\mathfrak{I}\models \verb!x a r:Proof.!$ iff \verb!x! is the proof step which leads to the proven result.
\item $\mathfrak{I}\models \verb!x a r:Parsing.!$ iff \verb!x! is a parsed axiom. 
\item $\mathfrak{I}\models \verb!x a r:Conjunction.!$ iff \verb!x! is a conjunction introduction.
\item $\mathfrak{I}\models \verb!x a r:Inference.!$ iff \verb!x! is a generalized modus ponens.
\item $\mathfrak{I}\models \verb!x a r:Extraction.!$ iff \verb!x! is a conjunction elimination.
\end{itemize}
 \item 
 Proof predicates:
\begin{itemize}
\item $\mathfrak{I}\models \verb!x r:gives {!f\verb!}!$. iff $f\in F$ is the formula obtained by applying \verb!x!.
\item $\mathfrak{I}\models \verb!x r:source u!$. iff \verb!x! is a parsed axiom and $\verb!u!\in U$ is the \uri of the 
parsed axiom's
source. 
\item $\mathfrak{I}\models \verb!x r:component y!$. iff \verb!x! is a conjunction introduction and $\verb!y!$ is a proof step which gives one of its components.
\item $\mathfrak{I}\models \verb!x r:rule y.!$ iff \verb!x! is a generalized modus ponens and $\verb!y!$ is the proof step which leads to the applied implication.
\item $\mathfrak{I}\models \verb!x r:evidence (y!_1,\ldots,\verb!y!_n\verb!)!$. iff \verb!x! is a generalized modus ponens and $\verb!y!_1,\ldots, \verb!y!_n$ 
are the proof steps which lead to the formulas used for the unification with the antecedent of the implication.
\item $\mathfrak{I}\models \verb!x r:because y!$. iff \verb!x! is a conjunction elimination and $\verb!y!$ is the proof step which yields the to-be-eliminated conjunction.
\end{itemize}
\item
Substitutions:
\begin{itemize}
 \item $\mathfrak{I}\models \verb!x r:binding z!_1.$ iff \verb!x! includes a substitution $\verb!z!_1$.
 \item $\mathfrak{I}\models \verb!z!_1 \verb! r:variable z!_2.$ iff $\verb!z!_1$ is a substitution whose domain is $\{\verb!z!_2\}$.
 \item $\mathfrak{I}\models \verb!z!_1 \verb! r:boundTo z!_3.$ iff $\verb!z!_1$ is a substitution whose range is $\{\verb!z!_3\}$.
\end{itemize}

\end{enumerate}
\end{definition}

To produce a proof for an API composition problem, the reasoner needs to be aware of all formulas at its disposal (in our case $H \cup R \cup B$) and of the goal 
which it is expected to prove.
The latter is given to the reasoner as the consequence of a~\textit{filter rule}
$\verb!{!f\verb!} => {!g\verb!}.!$
This triggers the reasoner to prove an instance of~$f$ and in case of success,
return each provable ground instance of $g$ if possible,
or a provable instance containing existentials otherwise.
For brevity, not all reasoners display every proof step in a proof:
especially conjunction elimination and introduction are often omitted.
However, to the best of our knowledge,
all reasoners' proofs contain all applications of $\verb!r:Inference!$ leading to a goal~$g$,
which allows us to measure a~proof's length by counting applications of the generalized modus ponens.

 \begin{lstlisting}[
  float=t,
  caption={The initial knowledge of the agent \emph{(agent\_knowledge.n3)}},
  label=lst:Knowledge]
§\textcolor{gray}{@prefix dbpedia: <http://dbpedia.org/resource/>.}§

<lena.jpg> a dbpedia:Image.
\end{lstlisting}

\begin{lstlisting}[
  float=t,
  caption={A filter rule expressing the agent's goal \emph{(agent\_goal.n3)}},
  label=lst:Goal]
§\textcolor{gray}{@prefix dbpedia-owl: <http://dbpedia.org/ontology/>.}§

{ <lena.jpg> dbpedia-owl:thumbnail ?thumbnail. }
=>
{ <lena.jpg> dbpedia-owl:thumbnail ?thumbnail. }.
\end{lstlisting}

\begin{lstlisting}[
  float=p,
  caption={Example hypermedia API composition proof},
  label=lst:Proof]
§\textcolor{gray}{@prefix dbpedia: <http://dbpedia.org/resource/>.}§
§\textcolor{gray}{@prefix dbpedia-owl: <http://dbpedia.org/ontology/>.}§
§\textcolor{gray}{@prefix ex: <http://example.org/image\#>.}§
§\textcolor{gray}{@prefix http: <http://www.w3.org/2011/http\#>.}§
§\textcolor{gray}{@prefix n3: <http://www.w3.org/2004/06/rei\#>.}§
§\textcolor{gray}{@prefix r: <http://www.w3.org/2000/10/swap/reason\#>.}§

§\bfseries<\#proof> a r:Proof, r:Conjunction;\label{ln:Proof}§
  r:component <#lemma1>;
  r:gives { <lena.jpg> dbpedia-owl:thumbnail _:sk3. }.

§\bfseries<\#lemma1> a r:Inference;\label{ln:Lemma1}§
  r:gives { <lena.jpg> dbpedia-owl:thumbnail _:sk3. };
  r:evidence (<#lemma2>);
  r:binding [ r:variable [ n3:uri "var#x0"]; r:boundTo [ n3:nodeId "_:sk3"]];
  r:rule <#lemma7>.
  
§\bfseries<\#lemma2> a r:Inference;\label{ln:Lemma2}§
  r:gives { _:sk4 http:methodName "GET".
            _:sk4 http:requestURI _:sk3.
            _:sk4 http:resp _:sk5.
            _:sk5 http:body _:sk3.
            <lena.jpg> dbpedia-owl:thumbnail _:sk3.
            _:sk3 a dbpedia:Image.
            _:sk3 dbpedia-owl:height 80.0. };
  r:evidence (<#lemma3>);
  r:binding [ r:variable [ n3:uri "var#x0"]; r:boundTo [ n3:uri "lena.jpg"]];
  r:binding [ r:variable [ n3:uri "var#x1"]; r:boundTo [ n3:nodeId "_:sk3"]];
  r:binding [ r:variable [ n3:uri "var#x2"]; r:boundTo [ n3:nodeId "_:sk4"]];
  r:binding [ r:variable [ n3:uri "var#x3"]; r:boundTo [ n3:nodeId "_:sk5"]];
  r:rule <#lemma4>.
  
§\bfseries<\#lemma3> a r:Inference;\label{ln:Lemma3}§
  r:gives { _:sk0 http:methodName "POST".
            _:sk0 http:requestURI "/images/".
            _:sk0 http:body <lena.jpg>.
            _:sk0 http:resp _:sk1.
            _:sk1 http:body <lena.jpg>.
            <lena.jpg> ex:comments _:sk2.
            <lena.jpg> ex:smallThumbnail _:sk3. };
  r:binding [ r:variable [ n3:uri "var#x0"]; r:boundTo [ n3:uri "lena.jpg"]];
  r:binding [ r:variable [ n3:uri "var#x1"]; r:boundTo [ n3:nodeId "_:sk0"]];
  r:binding [ r:variable [ n3:uri "var#x2"]; r:boundTo [ n3:nodeId "_:sk1"]];
  r:binding [ r:variable [ n3:uri "var#x3"]; r:boundTo [ n3:nodeId "_:sk2"]];
  r:binding [ r:variable [ n3:uri "var#x4"]; r:boundTo [ n3:nodeId "_:sk3"]];   
  r:evidence (<#lemma6>);
  r:rule <#lemma5>.
  
§\bfseries<\#lemma4> a r:Extraction;§§\label{ln:Lemma4}§ r:because [ a r:Parsing; r:source <desc_thumbnail> ].
§\bfseries<\#lemma5> a r:Extraction;§§\label{ln:Lemma5}§ r:because [ a r:Parsing; r:source <desc_images> ].
§\bfseries<\#lemma6> a r:Extraction;§§\label{ln:Lemma6}§ r:because [ a r:Parsing; r:source <agent_knowledge> ].
§\bfseries<\#lemma7> a r:Extraction;§§\label{ln:Lemma7}§ r:because [ a r:Parsing; r:source <agent_goal> ].
\end{lstlisting}

\newcommand\lineref[1]{[\textit{line~\ref{#1}}]}
\newcommand\linesref[2]{[\textit{lines~\ref{#1}--\ref{#2}}]}


\subsection{Discussion of a~Composition Proof Example}\label{example}
\paragraph{\bfseries Overview}
\Cref{lst:Proof} displays an example of a hypermedia API composition proof.
In addition to the two \restdesc descriptions from \cref{lst:Thumbnail,lst:Images}
($R$ in an API composition problem),
it involves the goal file \cref{lst:Goal} (\texttt{\{g\} => \{g\}.})
and the input file \cref{lst:Knowledge} (the initial knowledge $H$).
The proof, generated by EYE~\cite{eyepaper},
explains how, given an image, a~thumbnail of this image can be obtained.
Its main element is \verb!r:Proof!~\lineref{ln:Proof},
which is a~conjunction of different components, indicated by \verb!r:component!.
In this example, there is only one,
but there can be multiple.
The conclusion of the proof
(object of the \verb!r:gives! relation)
is the triple
\verb!<lena.jpg> dbpedia-owl:thumbnail _:sk3 !%
(with \verb!_:sk3! an existential variable).
This captures the fact that \verb!lena.jpg! has a~certain thumbnail,
which is the consequent of the filter rule in \cref{lst:Goal}.

\vspace{-1em}

\paragraph{\bfseries Extractions and inferences}
EYE represents \verb!r:extraction!s and \verb!r:inference!s as lemmata,
which can be reused in different proof steps.
The extractions, Lemmata~4 to~7~\linesref{ln:Lemma4}{ln:Lemma7}, are fairly trivial as none of the underlying formulas is a real conjunction and will therefore not be discussed in detail.
In this proof, they just give the formulas specified in \cref{lst:Thumbnail} to 4 with renamed variables.
The inferences describe the actual reasoning carried out
and thus merit a~closer inspection.
We will follow the path backwards from the proof's conclusion,
tracing back inferences until we arrive at the atomic facts
that are the starting point of the proof.

\vspace{-1em}

\paragraph{\bfseries Filter rule}
The justification for the conclusion
consists in this case of a~single component, namely Lemma~1~\lineref{ln:Lemma1}.
This lemma is an application of the modus ponens using the implication of Lemma~7,
which is the file \verb!agent_goal.n3! shown in \cref{lst:Goal}.
This rule has been instantiated 
according to the variable bindings indicated by \verb!r:binding!:
here, the variable \verb!?thumbnail! (\verb!var#x0!)
is bound to the existential variable \verb!_:sk3!.
The unification uses a part of the result of Lemma~2 as indicated by \verb!r:evidence!, the performed conjunction elimination is not listed in this proof.
\sloppy Substituting \verb!_:sk3! for \verb!?thumbnail! in \cref{lst:Goal}
indeed gives the desired conclusion
\verb!<lena.jpg> dbpedia-owl:thumbnail _:sk3!,
which contributes to the final~result.

\vspace{-1em}

\paragraph{\bfseries Thumbnail operation}
Lemma~2~\lineref{ln:Lemma2} is another \verb!r:inference!,
this time applying the rule from Lemma~4,
which is the \restdesc description in \cref{lst:Thumbnail}
that explains what happens when an image is uploaded. 
The instantiation is again detailed with \verb!r:binding! statements.
The \verb!?image! variable (\verb!var#x0!) is bound to \verb!lena.jpg!,
the \verb!?thumbnail! variable (\verb!var#1!) to the existential \verb!_:sk3!.
This is only possible because a~ statement
\verb!<lena.jpg> ex:smallThumbnail _:sk3! is proven; 
its derivation will be detailed shortly.
The other variables \verb!var#x2! and \verb!var#x3!
refer to the request and response resources in the consequent of \cref{lst:Thumbnail}
and are instantiated with new existentials (\verb!_:sk4! and \verb!_:sk5!).
This lemma is in turn made possible by another one.

\vspace{-1em}

\paragraph{\bfseries Upload operation}
Lemma~3~\lineref{ln:Lemma3} explains the derivation of the triple
that is a~necessary condition for Lemma~2:
\verb!<lena.jpg> ex:smallThumbnail _:sk3!.
The rule used for the inference is that
from the upload action in \cref{lst:Images} (Lemma~5),
instantiated with the initial knowledge of the agent
that it has access to an image (\cref{lst:Knowledge} / Lemma~6).
Substituting this knowledge into the rule
by binding \verb!?image! (\verb!var#x0!)
 to \verb!lena.jpg!,
gives the triples of the instantiated consequent,
as shown by the \verb!r:gives! predicate.
Note in particular the newly created existential \verb!_:sk3!;
this entails the triple \verb!<lena.jpg> ex:smallThumbnail _:sk3!
which is needed for Lemma~2.
Lemma~3 itself is justified by Lemma~6 (\cref{lst:Knowledge}),
which is a~simple extraction that stands on itself.
Hence, we have reached the starting proof's starting~point.

\vspace{-1em}

\paragraph{\bfseries From start to end}
As a~summary, we will briefly follow the proof in a~forward way.
Extracted from the background knowledge in \cref{lst:Knowledge},
the triple \verb!<lena.jpg> a dbpedia:Image!
triggers the image upload rule from \cref{lst:Thumbnail},
which yields \verb!<lena.jpg> dbpedia-owl:thumbnail _:sk3!
for some existential variable \verb!_:sk3!.
This triple in turn triggers the thumbnail rule from \cref{lst:Thumbnail},
which yields \verb!<lena.jpg> ex:smallThumbnail _:sk3!.
Finally, this is the input for the filter rule from \cref{lst:Goal},
which yields the final result of the proof:
the image has some thumbnail.
So given the background knowledge
and the two hypermedia API descriptions,
this proof verifies that the goal follows from them.

\subsection{Hypermedia API Operations inside a~Proof}
The proof in \cref{lst:Proof} is special
in the sense that some of its implication rules,
namely \cref{lst:Thumbnail,lst:Images},
are actually hypermedia API descriptions.
That means they do not fulfill an actual ontological implication.
Instead, they convey dynamic information.
Therefore, those steps in the proof
can be interpreted as \http requests that should be performed
in order to achieve the desired result.
This proof is indeed a~pre-proof:
it is valid under the assumption
that the described \http requests will behave as expected,
which can never be guaranteed on an environment such as the Internet.
The instantiation of a~hypermedia API description
turns it into the description of a~concrete API operation.
For instance, Lemma~3~\lineref{ln:Lemma3} contains the following operation:
\begin{Verbatim}
    _:sk0 http:methodName "POST".
    _:sk0 http:requestURI "/images/".
    _:sk0 http:body <lena.jpg>.
    _:sk0 http:resp _:sk1.
    _:sk1 http:body <lena.jpg>.
    <lena.jpg> ex:comments _:sk2.
    <lena.jpg> ex:smallThumbnail _:sk3.
\end{Verbatim}
This instructs an agent to perform
an \http \verb!POST! request (\verb!_:sk0!) to \verb!/images/!
with \verb!lena.jpg! as the request body.
This request will return a~response (\verb!_:sk1!)
with a~representation of \verb!lena.jpg!,
which will contain \verb!ex:comments! and \verb!ex:smallThumbnail! links.
Since this is a~hypermedia API,
the link targets are not yet known at this stage.
Therefore, they are represented by generated existential variables \verb!_:sk2! and \verb!_:sk3!.
At runtime, the \http request will return the actual link targets,
but at the pre-proof stage,
it suffices to know that some target for those links will exist.

Indeed, the outcome of this rule---the fact that \emph{some} \verb!smallThumbnail! links will exist---%
serves as input for the next API operation in Lemma~2 \lineref{ln:Lemma2}.
The consequent of this rule is instantiated as follows:
\begin{Verbatim}
    _:sk4 http:methodName "GET".
    _:sk4 http:requestURI _:sk3.
    _:sk4 http:resp _:sk5.
    _:sk5 http:body _:sk3.
    <lena.jpg> dbpedia-owl:thumbnail _:sk3.
    _:sk3 a dbpedia:Image.
    _:sk3 dbpedia-owl:height 80.0
\end{Verbatim}
This describes a~\verb!GET! request (\verb!_:sk4!) to the URL \verb!_:sk3!,
which will return a~representation of a~thumbnail that is 80~pixels high.
This request is interesting because it is \emph{incomplete}:
\verb!_:sk3! is not a~concrete URL that can be filled in. 
However, this identifier is the same variable as the one in Lemma~3,
so this description essentially states
that whatever will be the target of the \verb!smallThumbnail! link in the previous \verb!POST! request
should be the URL of the present \verb!GET! request.
The existential variables thus serve as placeholders
for values that will be the result of actual API operations.

While the proof above is a~pre-proof,
a~proper post-proof can be obtained by actually executing the \verb!POST! \http request,
which has all values necessary for execution
(as opposed to the \verb!GET! request where the URL is still undetermined).
This execution will result in a~concrete value
for the \verb!comments! and \verb!smallThumbnail! link placeholders \verb!_:sk2! and \verb!_:sk3!.
They lead to a~proper post-proof that uses these concrete values,
and hence that proof does not need the assumption that the \verb!POST! request will execute successfully
(because evidence shows it did).
\Cref{fig:UmlDiagram} shows the UML sequence diagram
of an example interaction between the client and the server.

\begin{figure}
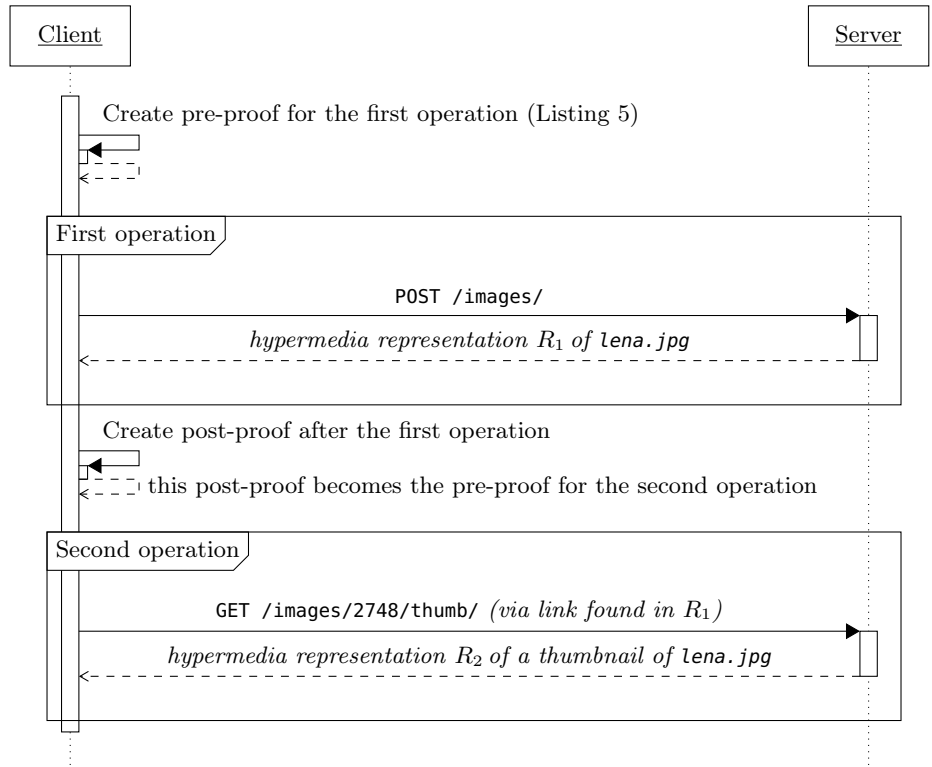

  \begin{sequencediagram}
    \newthread[white]{c}{Client}
    \newinst[9]{s}{Server}

    \begin{callself}{c}{Create pre-proof for the first operation (\cref{lst:Proof})}{}\end{callself}

    \begin{sdblock}{First operation}{}
      \begin{call}{c}{\texttt{POST /images/}}{s}
        {\textit{hypermedia representation~$R_1\thinspace$of \texttt{lena.jpg}}}
      \end{call}
    \end{sdblock}

    \begin{callself}{c}{Create post-proof after the first operation}
                       {this post-proof becomes the pre-proof for the second operation}\end{callself}

    \begin{sdblock}{Second operation}{}
      \begin{call}{c}{\texttt{GET /images/2748/thumb/} \textit{(via link found in~$R_1$)}}{s}
        {\textit{hypermedia representation~$R_2\thinspace$of a~thumbnail of \texttt{lena.jpg}}}
      \end{call}
    \end{sdblock}
  \end{sequencediagram}
  \caption{
    The \emph{a~posteriori} UML sequence diagram of an example execution
    shows how the second operation happened through a~link in the first operation.
  }
  \label{fig:UmlDiagram}
\end{figure}

As stated in its definition,
a~pre-proof implicitly assumes that each API will indeed deliver the functionality
as stated in its \restdesc description.
The proof thus only holds under that assumption.
For example, if a~power outage occurs during the calculation of the aspect ratio,
the placeholder will not be instantiated with an actual value during the execution,
which can pose a~threat to subsequent hypermedia API operations that depend on this value.
However, the failure of a~single API operation
does not necessarily imply the intended result cannot be achieved.
Rather, it means the assumption of the pre-proof was invalid
and an alternative pre-proof---a~new hypermedia API composition---should be created,
starting from the current application state.
Such a~pragmatic approach to proofs containing hypermedia API operations is unavoidable:
no matter how low the probability of a~certain operation to fail,
failures can never be eliminated.
Therefore, pragmatism ensures that planning in advance is possible.
Each proof should be stored along with its assumptions
in order to understand the context it which it can be used.

\section{Hypermedia-driven Composition Generation and Execution}
\label{sec:Composition}
In contrast to fully plan-based methods,
the steps in the composition obtained through reasoner-based composition of hypermedia APIs
are not executed blindly.
Instead, the interaction is driven by the hypermedia responses from the server;
the composition in the proof only serves as \emph{guidance} for the client,
and as a~guarantee (to the extent possible) that the desired goal can be reached.
The composition that starts from the current state
helps an agent decide what its next step towards that goal should be.
Once this step has been taken, the rest of the pre-proof is discarded
because it is based on outdated information.
After the request,
the state is augmented with the information inside the server's response.
This new state becomes the input for a~new pre-proof
that takes into account the actual situation,
instead of the expected (and incomplete) values from the hypermedia API description.
In this section, we will detail this iterative composition generation
and execution.

\subsection{Goal-oriented Composition Generation}
Creating a~composition that satisfies a~goal comes down
to generating a~proof that supports the goal.
Inside this proof, the necessary hypermedia API operations will be incorporated as instantiated rules.
Proof-based composition generation,
unlike other composition techniques,
requires no composition-specific tools or algorithms.
A~generic reasoner that supports the rule language in which the hypermedia APIs are described
is capable of generating a~proof containing the composition.
For example, since \restdesc descriptions are expressed in the \nthree language,
compositions of hypermedia APIs described with \restdesc
can be performed by any \nthree reasoner with proof support.
The fact that proof-based composition can be performed by existing reasoners
is an advantage in itself,
because no new software has to be implemented and tested.
Furthermore, this offers the following benefits.
\begin{description}
\item[Incorporation of external knowledge]
  Existing RDF knowledge can directly be incorporated into the composition process.
  Whereas composition algorithms that are specifically tailored to certain description models
  usually operate on closed worlds,
  generic Semantic Web reasoners are built to incorporate knowledge from various sources.
  For example, existing OWL and RDF ontologies can be used
  to compose hypermedia APIs described with different vocabularies.
\item[Evolution of reasoners]
  Many implementations of reasoners exist
  and they continue to be updated to allow enhanced performance and possibilities.
  The proof-based composition method directly benefits from these innovations.
  This also counters the problem that many single-purpose composition algorithms
  are seldom updated after their creation because they are so specific.
\item[Independent validation]
  When dealing with proof and trust on the Web,
  it is especially important that the validation can happen by an independent party.
  Since different reasoners and validators exist,
  the composition proof can be validated independently.
  This contrasts with other composition approaches,
  whose algorithms have to be trusted.
\end{description}

In order to make a~reasoner generate a~pre-proof of a~composition,
it must be invoked with the initial state, the available hypermedia APIs, and the desired goal.
Here, we will examine the case for \nthree reasoners
and hypermedia APIs described with \restdesc\ \nthree rules,
but the principle of proof-based composition is generalizable
to all families of inference rules.

The hypermedia API descriptions include all those APIs available to the client.
In practice, the number of supplied available hypermedia APIs
would be substantially higher than the number of APIs in the resulting composition.
The background knowledge can, for example, consist of ontologies and business rules.
The reasoner will try to infer the goal state,
asserting the other inputs as part of the ground truth.
The initial state and background knowledge should correspond to reality,
regardless of the results of the actual execution,
provided the descriptions are accurate.
In contrast, the API description rules only hold under the assumption of successful execution,
due to the nature of the pre-proof.


If the reasoner can infer the goal state given the ground truth,
we can conclude that a~composition \emph{exists}.
To obtain the details of the composition,
the reasoner must return the proof of the inference,
i.e., the data and rules applied to achieve the goal.
Inside this proof,
there will be placeholders for return values by the server
that are unknown at design-time.
The proof will be structured as in \cref{lst:Proof},
where the initial state was \cref{lst:Knowledge},
the goal state \cref{lst:Goal},
and the descriptions \cref{lst:Thumbnail,lst:Images}.
No background knowledge was needed,
but it could have been useful for instance
if the image of the initial state was described in different ontology,
in which case the conversion to the \dbpedia ontology would be necessary.

\subsection{Hypermedia-driven Execution}
\label{subsec:Execution}
In order to achieve a~certain goal in a~hypermedia-driven way,
the following process steps can be followed.
\begin{definition}[Pragmatic proof algorithm]
\label{def:PPAlgorithm}
Given an API composition problem with an initial state $H$, goal $g$,
description formulas $R$ and background knowledge $B$,
we define the pragmatic proof algorithm as follows:
\begin{enumerate}
\item
  Start an \nthree reasoner to generate a pre-proof for $(R, g, H, B)$. 
  \begin{enumerate}
  \item If the reasoner is not able to generate a proof, halt with \emph{failure}.
  \item Else scan the pre-proof for applications of rules of $R$, set the number of these applications to $n_{\text{pre}}$.
  \end{enumerate} 
  \item Check $n_{\text{pre}}$:
  \begin{enumerate}
  \item If $n_{\text{pre}}=0$
  , halt with \emph{success}.
  \item Else 
  continue with Step~3.
  \end{enumerate}
%
\item
  Out of the pre-proof, select a sufficiently specified \http request description which is part of 
  the application of a rule $r\in R$.
\item  Execute the described HTTP request
  and parse the (possibly empty) server response to a set of ground formulas $G$. 
\item  Invoke the reasoner with the new 
  API composition problem $(R, g, H\cup G, B)$ to produce a post-proof.
 \item  Determine $n_{\text{post}}$:
 \begin{enumerate}

 \item If the reasoner was not able to generate a~proof,
 set $n_{\text{post}}:=n_{\text{pre}}.$


 \item Else scan the proof for the number of inference steps which are using rules from $R$ and set this number of steps to $n_{\text{post}}$.  
  \end{enumerate}
  \item Compare $n_{\text{post}}$ with $n_{\text{pre}}$:
 \begin{enumerate}
 \item If $n_{\text{post}} \geq n_{\text{pre}}$ 
	go back to Step~1 with the new API composition problem $(R\setminus \{ r\}, g, H, B)$. 
 \item If $n_{\text{post}} < n_{\text{pre}}$, the post-proof
 can be used as the next pre-proof.
       Set $n_{\text{pre}}:=n_{\text{post}}$ and continue with Step~2.
 \end{enumerate}
  
\end{enumerate}
\end{definition}

Before having a more theoretical look at the results of this algorithm,
let us run through a~possible execution
of the composition example introduced previously.

\newcommand\stepref[1]{\emph{(Step~#1)}}
\begin{itemize}
  \item \stepref{1} Given the background knowledge, initial state, and goal,
        the reasoner generates the pre-proof from \cref{lst:Proof},
        which contains $n_{\text{pre}} = 2$ API operations.
  \item \stepref{2} $n_{\text{pre}} \neq 0$, so continue with Step~3.
  \item \stepref{3} The HTTP request to upload the image
        is the only one that is sufficiently instantiated,
        so it is selected.
  \item \stepref{4} Execute the \http request
        by posting the image to \verb!/images/!,
        and retrieve a~hypermedia response.
        Inside this hypermedia response,
        there is a~\verb!comments! link to \verb!/comments/about/images/37!
        and a~\verb!smallThumbnail! link to \verb!/images/37/thumb/!.
        They are added to~$G$.
  \item \stepref{5} A~post-proof is produced from the new state,
        revealing that the goal can now be completed with one API operation.
        Indeed, only an \http \verb!GET! request to \verb!/image/37/thumb! is needed.
  \item \stepref{6} The above means that $n_{\text{post}} = 1$.
  \item \stepref{7} $n_{\text{post}} = 1 < n_{\text{pre}} = 2$,
        so set $n_{\text{pre}} := 1$ and continue with Step~2.
  \item \stepref{2} $n_{\text{pre}} = 1 \neq 0$, so continue with Step~3.
  \item \stepref{3} Select the only remaining HTTP request in the pre-proof.
  \item \stepref{4} Execute the \verb!GET! request to \verb!/image/37/thumb!
        and thereby obtain a~representation of the thumbnail of the image.
  \item \stepref{5} Generate the post-proof;
        it consists entirely of data
        as the necessary information to reach the goal has been obtained.
  \item \stepref{6} No rules of~$R$ are applied, so $n_{\text{post}} = 0$.
  \item \stepref{7} $n_{\text{post}} = 0 < n_{\text{pre}} = 1$,
        so set $n_{\text{pre}} := 0$ and continue with Step~2.
  \item \stepref{2} $n_{\text{pre}} = 0$, so halt with \textit{success}.
\end{itemize}

\noindent
This example shows how the proof guides the process,
but hypermedia drives the interaction.
For instance, the \URL needed for the \verb!GET! request was not hard-coded:
it was obtained as a~hypermedia control from the server.
This means that, even if the server changes its internal \URL structure
or the layout of the representation,
the interaction can still take place.
The client needs the RESTdesc descriptions to find out whether the complex goal is possible
and what first steps it should take.
Otherwise, it would have no way of knowing
that the upload of an image results in a~link to the thumbnail.
However, once this expectation is there,
the client navigates through hypermedia.
We can compare this to driving with a~map:
the map gives the overall picture,
but the actual wayfinding happens based on the actual roads and scenery
when somebody undertakes the journey.

There are several reasons, why the situations in 1(a) or 6(a) can occur:
the reasoner could have a~technical problem, it could detect inconsistencies
(a fulfilled antecedent of a rule with false in the consequent), it could simply be that 
there is no instance $g'$ of $g$ which is a logical consequence of $H \cup R\cup B$, but even
if such a $g'$ exists the reasoning problem is undecidable.
This is because RESTdesc descriptions are rules with new existential variables in the consequence,
which makes the problem in general undecidable, as discussed by \citeN{Baget}.

However, we can show that the following holds:
\begin{theorem}
Given an execution of the algorithm in \cref{def:PPAlgorithm}
that requires $n$~executions of the reasoner (in Steps 1 and~5).
If all $n$~reasoning runs terminate, the algorithm terminates as well.
Furthermore, if the algorithm halts with \emph{success},
its output is a~ground instance of the goal state.

\end{theorem}

\begin{proof}
We first show termination:
if the algorithm does not terminate, it especially never reaches the Steps 1(a) and 2(a) and
Steps 1 and 2 always result in option (b). 
All formulas in $H\cup B$ are either ground formulas
or simple rule formulas 
which do not contain existentials 
and which fulfill Lemma~\ref{lemma:Reasoning}.
Therefore no atomic formula or conjunction of atomic formulas
which can be obtained by applying the proof steps from Definition \ref{proofsteps}
on $H\cup B$
contains universals 
or existentials.
If for a pre-proof $n_{\text{pre}}>0$ holds, it must 
by Corollary~\ref{corollary} 
contain at least one
sufficiently specified \http request description. So, Step 3 and the following Steps~4--6 can always be executed.
Step 7(a) reduces the set of \restdesc descriptions, it can only be performed $|R|$ times.
Starting from a fixed pre-proof $\text{pre}_0$, Step~7(b) can only be applied $n_{\text{pre}_0}$ times.
Thus, the algorithm terminates.

As every operation which changes the API composition problem for the pre-proof to be checked in Step 2, preserves the syntactic properties of the 
sets of formulas involved, it is enough to show that the result of every pre-proof of an API composition problem which does not contain 
applications of \restdesc rules is ground. 
This follows by the same arguments as above. As $H\cup B \cup \{ \texttt{\underline{\{}g\underline{\}} => \underline{\{}g\underline{\}}.}\}$ 
only contains ground formulas and rules which fulfill
the conditions of Lemma~\ref{lemma:Reasoning},
it is not possible to derive atomic formulas or conjunctions of atomic formulas which contain existentials or universals 
thus the result of the proof is ground.
\end{proof}

\section{Semi-Automated Description Generation}
\label{sec:Generation}
One of the bottlenecks with traditional description-based methods of Web APIs and services
is that these descriptions have to be created, which is mostly a~manual task.
Consequently, without a~method that facilitates the creation of such descriptions,
the overall concept of description-based composition and execution
might not successfully transition from theory to practice.
This section explains how RESTdesc descriptions can be created
by a~computer-assisted process.

We define \emph{semi-automated RESTdesc description generation}
as a~process that takes as input a~series of HTTP requests and responses
performed by one or more persons,
and provides as output skeletons for RESTdesc descriptions,
which a~user can then further refine to a~final description.
We rely on the fact that RESTdesc descriptions
capture the \emph{expectations} of resources.
For example, \cref{lst:Images} captures the expectation that
the upload of an image will result in links to comments and a~thumbnail.
RESTdesc thus describes the generic hypermedia controls
that will be available on such image resources.

\begin{lstlisting}[
  float=t,
  caption={
    Example user interaction with a~Web API,
    used as input for the RESTdesc description generation process.
  },
  label=lst:Interaction]
§\textbf{HTTP request~1: POST to /images/ with image1.jpg as body}§
§\textbf{HTTP response~1:}§
@prefix ex: <http://example.org/image#>.
@prefix dbpedia: <http://dbpedia.org/resource/>.

</images/24> a dbpedia:Image.
             ex:comments </images/24/comments>.
             ex:smallThumbnail </images/24/thumbnail>.


§\textbf{HTTP request~2: GET to /images/24/thumbnail}§
§\textbf{HTTP response~2:}§
@prefix dbpedia: <http://dbpedia.org/resource/>.
@prefix dbpedia-owl: <http://dbpedia.org/ontology/>.

</images/24> dbpedia-owl:thumbnail </images/24/thumbnail>.
</images/24/thumbnail> a dbpedia:Image;
                       dbpedia-owl:height 80.0.


§\textbf{HTTP request~3: POST to /images/ with image2.jpg as body}§
§\textbf{HTTP response~3:}§
@prefix ex: <http://example.org/image#>.
@prefix dbpedia: <http://dbpedia.org/resource/>.

</images/25> a dbpedia:Image.
             ex:comments </images/25/comments>.
             ex:smallThumbnail </images/25/thumbnail>.


§\textbf{HTTP request~4: GET to /images/25/thumbnail}§
§\textbf{HTTP response~4:}§
@prefix dbpedia: <http://dbpedia.org/resource/>.
@prefix dbpedia-owl: <http://dbpedia.org/ontology/>.

</images/25> dbpedia-owl:thumbnail </images/25/thumbnail>.
</images/25/thumbnail> a dbpedia:Image;
                       dbpedia-owl:height 80.0.
\end{lstlisting}

The idea behind the process is
to extract the hypermedia controls from a~series of requests performed by people.
We will describe the strategy using an exemplary series of interactions,
displayed in \cref{lst:Interaction}.
To create this particular series,
a~user uploaded two images and obtained their thumbnails.

First, the process needs to identify the different kinds of steps and thus resources.
This happens by performing clustering on the HTTP responses with, for instance,
a~string similarity algorithm as distance function.
In this case, responses 1 and~3 are highly similar,
and so are responses 2 and~4.
Therefore, they are assigned into two different clusters.
Note that such clustering is especially realistic
because Web APIs typically generate responses based on templates,
which contributes to a~high structural similarity
for responses of the same kind that hence follow the same template.
User interaction can influence the clustering sensitivity,
and possibly manually change assignments.
Since the two clusters, corresponding to 2 types of operations,
are correct in this example, nothing needs to change.

Next, for each cluster,
the common elements in responses and their corresponding requests are identified.
In this case, the first cluster contains two \verb!POST! requests to \verb!/images/!,
and both responses contain some resource that is an image,
and has a~link to comments and images.
This allows the algorithm to produce a~skeleton as follows:
\begin{Verbatim}
{ ?object1 a :localFile. }
=>
{
  _:request http:methodName "POST";
            http:requestURI "/images/";
            http:body ?object1;
            http:resp [ http:body _:object2 ].
  _:object2 a dbpedia:Image;
            ex:comments _:object3;
            ex:smallThumbnail _:object4.
}
\end{Verbatim}
Note how this is very close to \cref{lst:Images}.
The user can still improve this description
by specifying that \verb!?object1! and \verb!_object2! are the same entity,
and that \verb!?object1! is a~\verb!dbpedia:Image! already in the antecedent.
After these changes, and possibly renaming variables and blank nodes for legibility,
the description is the same as that of \cref{lst:Images} and thus correct.

Identifying common elements in the second cluster
leads to the following skeleton:
\begin{Verbatim}
{}
=>
{
  _:request http:methodName "GET";
            http:requestURI _:object1;
            http:resp [ http:body _:object1 ].
  _:object2 dbpedia-owl:thumbnail _:object1.
  _:object1 a dbpedia:Image;
            dbpedia-owl:height 80.0.
}.
\end{Verbatim}
The fact that the request~URI and the body refer to the same entity (\verb!_:object1!)
can be deduced from the properties of the \verb!GET! method in the HTTP specification.
These properties do not apply to \verb!POST!,
which is why the previous skeleton could not make this assumption.
The current skeleton is still incomplete, since it does not have a~precondition.
More specifically, we need to obtain the request~URI somehow.
Given the sequence of the requests in \cref{lst:Interaction},
the process can detect that the concrete instances of \verb!_:object1!
(\verb!/images/24/thumbnail! and \verb!/images/25/thumbnail!)
were already mentioned in a~previous response body.
Hence, it can place the pattern containing those instances in the antecedent
and connect the components that had identical values with the same variable names:
\begin{Verbatim}
{ ?object2 ex:smallThumbnail ?object1. }
=>
{
  _:request http:methodName "GET";
            http:requestURI ?object1;
            http:resp [ http:body ?object1 ].
  ?object2 dbpedia-owl:thumbnail ?object1.
  ?object1 a dbpedia:Image;
           dbpedia-owl:height 80.0.
}.
\end{Verbatim}
The user can now optionally rename the variable placeholders
to obtain the exact same description as in \cref{lst:Thumbnail}.

This shows how descriptions such as \cref{lst:Thumbnail,lst:Images}
can be generated.
Note that this process is not fully automated, but human-assisted.
That is, it requires a~repeated sequence of human steps as input
in order to sufficiently cluster and generalize descriptions.
Furthermore, hints and corrections from users might be necessary,
as over- or undergeneralizations will occur inevitably in some cases.
Yet such an assisted process significantly decreases the burden
of full manual description.
The fact the RESTdesc descriptions focus on REST APIs facilitates the process:
it can make additional assumptions on the behavior of the uniform interface,
and hyperlinks from previous responses can be reused.

Most important for this assisted generation of descriptions
is the close relationship between N3 and RDF.
As we have shown above, the RDF triples in the Web API's responses
serve as a~direct prototype for the consequent of the generated skeletons.
For example, the triple
\verb!</images/24>! \verb!dbpedia-owl:smallThumbnail! \verb!</images/24/thumbnail>.!
directly leads to the inclusion of
\verb!_:object2! \verb!dbpedia-owl:smallThumbnail! \verb!_:object4!
in the first skeleton.
Not only does this simplify the generation process,
the connection between the generated description
and the original response is also apparent for users,
which makes it easier for them as well.

\section{Evaluation}
\label{sec:Evaluation}
\subsection{Composition Algorithm Benchmark}
This article discusses a~proof-based method to compose and execute hypermedia APIs.
Specific to this method, in contrast to traditional Web service composition methods,
is that the composition should be regenerated at each step.
This makes the feasibility of the approach depend
on whether the composition cost is within reasonable limits.
Therefore, an evaluation should assess whether
composition happens sufficiently fast
for realistic composition lengths
and in presence of a~realistic number of possible hypermedia APIs
that can be used in compositions.

To verify this, we developed a~benchmark framework%
\footnote{Available at \url{http://github.com/RubenVerborgh/RESTdesc-Composition-Benchmark}.}
for hypermedia API composition,
consisting of two main components:
\begin{description}
\item[a~hypermedia API description generator]\hspace{-1ex},
which deterministically generates single- or multi-connected chains
of example hypermedia API descriptions with a~chosen length,
specifically tailored to enable compositions;
\item [an automated benchmarker]\hspace{-1ex},
testing how well a~reasoner performs
on creating proofs for compositions of varying lengths and complexity.
\end{description}

Below is one of the example descriptions,
generated by the tool with parameters \verb!3!~(description chain length)
and \verb!2!~(number of needed connections per description).
\begin{Verbatim}
    {
      ?a1 ex:rel2 ?b1.                 # Generated input conditions
      ?a2 ex:rel2 ?b2.
    }
    =>
    {
      _:request http:methodName "GET"; # Generated HTTP request
                http:requestURI ?a1;
                http:resp [ http:body ?b1 ].
      ?a1 ex:rel3 ?b1.                 # Generated output conditions
      ?a2 ex:rel3 ?b2.
    }.
\end{Verbatim}

The goal would be a~filter rule of the following form:
\begin{Verbatim}
    {
      ?a1 ex:rel3 ?b1.                 # Output conditions of the
      ?a2 ex:rel3 ?b2.                 # last API operation in the chain
    }
    =>
    {
      ?a1 ex:rel3 ?b1.
      ?a2 ex:rel3 ?b2.
    }.
\end{Verbatim}
Herein, the conditions in the antecedent
are only satisfiable by creating a~chain
from the first description towards the last.

As this example shows,
descriptions are structurally identical to RESTdesc descriptions of existing hypermedia APIs
and therefore representative examples.
Other descriptions will be generated
such that the input and output conditions can be matched,
so the composition algorithm can form chains
(in this example, with links to two previous descriptions).

\subsection{Parameters and Measurements}
The main parameters that determine the difficulty of generating a~composition are:
\begin{enumerate}
\item the \textbf{number} of API operations in the resulting composition:~$n$
\item the number of \textbf{dependencies} between hypermedia API operations in the composition:~$d$
\item the \textbf{total} number of hypermedia APIs supplied to the reasoner
      (not all necessarily part of the composition):~$t$
\end{enumerate}

To measure the influence of the first two parameters,
we test the generation of a~composition with a~resulting length of~$n$
and with $d$~dependencies between each hypermedia API operation,
where $n$ ranges from 2 to~1,024 and $d$ from 1 to~3.
These ranges have been chosen such that their upper bounds exceed
those of compositions for regular use cases,
which typically involve only a~few API operations
with few interdependencies.
The goal is therefore whether performance is acceptable for small values---%
success on larger values comes as an added bonus.

To test the third parameter~$t$,
we will keep $n$ and $d$ fixed at 32 and~1 respectively
(\textit{i.e.,} already a~large composition)
and add a~number of \emph{dummy} APIs~($n'=t-n$) that can be composed with the other APIs,
but are not needed in the resulting composition.
It is important to understand that most real-world scenarios will be a~mixture
of the above situations:
compositions are generally graphs with a~varying number of dependencies,
created in presence of a~non-negligible number of descriptions that are irrelevant to the composition under construction.
Therefore, by measuring these aspects independently,
we can predict the performance in those situations.

The measurements have been split in
\emph{parsing}, \emph{reasoning}, and \emph{total} times.
Parsing represents the time during which the reasoner internalizes the input
into an in-memory representation.
This was measured by presenting the inputs to the reasoner,
without asking for any operation to be performed on them.
Since the parsing step can often be cached and reused in subsequent iterations,
it is worthwhile evaluating the actual reasoning time separately.
Parsing and reasoning together make up for the total time.

\subsection{Results}

The benchmark was executed on one 2.4~GHz core of an Intel Xeon processor
on Ubuntu Server 12.04.
The results are summarized below;
full results are available
at \url{http://github.com/RubenVerborgh/RESTdesc-Composition-Benchmark-Results}.

\subsubsection{EYE reasoner}
\Cref{tbl:EYE1,tbl:EYE2} show the benchmark results achieved by the EYE reasoner~\cite{eyepaper},
version 2014-09-30 on SWI-Prolog 6.6.6.
The results in the first column teach us that
starting the reasoner introduces an overhead of $\approx 40$~ms.
This includes process starting costs, which are highly machine-dependent.

Inspecting \cref{tbl:EYE1} from left to right,
we see the reasoning time increases with the composition length~$n$
and remains limited to a~few hundred milliseconds in almost all cases.
The absolute increase in reasoning time for a~higher number of dependencies~$d$
never crosses 150~ms for small to medium values of~$n$,
but becomes larger for high~$n$.
\cref{tbl:EYE2} shows that the reasoning time hardly increases in presence of dummies.

\vspace{-1em}

\subsubsection{\cwm reasoner}
The same experiments have been performed with the \cwm reasoner~\cite{cwm},
whose results are shown in \cref{tbl:cwm1,tbl:cwm2}.
The \cwm reasoner is not as strongly performance-optimized as EYE,
which is clearly visible in the results.
Also, we were only able to test for values of~$n$ up to~256,
because out-of-memory errors appeared for large values.
Despite this fact, we still see acceptable results for small-to-medium-sized compositions.
We note a~higher start-up time of $\approx$~140~ms.
Reasoning time increases faster than linearly in~$n$,
which is also the case for increasing~$d$.
The presence of dummies bothers \cwm more than EYE,
and serious issues start to appear at $n'=$~512.

\begin{table}[t]
  \caption{The EYE reasoner manages to create even lengthy compositions in a~timely manner (\emph{average times of 50~trials}).}
  \label{tbl:EYE1}
  \fontsize{8pt}{9pt}\selectfont
  \begin{tabular*}{\textwidth}{@{\extracolsep{\fill}} r r r r r r r r r r}
    \hline\hline
    \hspace{5.2ex}\bf \#APIs~$n$ & \bf 4 & \bf 8 & \bf 16 & \bf 32 & \bf 64 & \bf 128 & \bf 256 & \bf 512 & \bf 1,024 \\
    \hline
    \bf{$d$=$1$~dep.}&&&&&&&&&\\
    parsing & 39~ms & 41~ms & 45~ms & 55~ms & 74~ms & 114~ms & 183~ms & 323~ms & 597~ms\\
    \em reasoning &  \em 16~ms &  \em 19~ms &  \em 33~ms &  \em 21~ms &  \em 43~ms &  \em 48~ms &  \em 108~ms &  \em 268~ms &  \em 903~ms\\
    total & 56~ms & 61~ms & 79~ms & 76~ms & 118~ms & 162~ms & 292~ms & 591~ms & 1,500~ms\\
    \hline
    \bf{$d$=$2$~dep.}&&&&&&&&&\\
    parsing & 39~ms & 44~ms & 59~ms & 61~ms & 90~ms & 126~ms & 212~ms & 383~ms & 728~ms\\
    \em reasoning &  \em 22~ms &  \em 60~ms &  \em 22~ms &  \em 56~ms &  \em 45~ms &  \em 113~ms &  \em 311~ms &  \em 1,044~ms &  \em 3,640~ms\\
    total & 61~ms & 105~ms & 82~ms & 117~ms & 136~ms & 239~ms & 524~ms & 1,427~ms & 4,369~ms\\
    \hline
    \bf{$d$=$3$~dep.}&&&&&&&&&\\
    parsing & 41~ms & 50~ms & 53~ms & 73~ms & 97~ms & 146~ms & 250~ms & 460~ms & 890~ms\\
    \em reasoning &  \em 29~ms &  \em 18~ms &  \em 39~ms &  \em 48~ms &  \em 105~ms &  \em 263~ms &  \em 787~ms &  \em 2,660~ms &  \em 9,105~ms\\
    total & 70~ms & 68~ms & 93~ms & 122~ms & 203~ms & 409~ms & 1,037~ms & 3,121~ms & 9,995~ms\\
    \hline\hline
  \end{tabular*}
\end{table}
\begin{table}[t]
  \caption{Even a~large number of dummies does not significantly disturb
  EYE's reasoning time (\emph{average times of 50~trials}).}
  \label{tbl:EYE2}
  \fontsize{8pt}{9pt}\selectfont
  \begin{tabular*}{\textwidth}{@{\extracolsep{\fill}} r r r r r r r r}
    \hline\hline
    \bf \# dummies~$n'$ & \bf 2,048 & \bf 4,096 & \bf 8,192 & \bf 16,384 & \bf 32,768 & \bf 65,536 & \bf 131,072\\
    \hline
    \bf{$n$=32, $d$=1~dep.}&&&&&&&\\
    parsing & 1,221~ms & 2,405~ms & 4,586~ms & 8,959~ms & 17,700~ms & 35,823~ms & 74,394~ms\\
    \em reasoning & \em 47~ms & \em 66~ms & \em 71~ms & \em 177~ms & \em 403~ms & \em 633~ms & \em 1,042~ms\\
    total & 1,268~ms & 2,471~ms & 4,657~ms & 9,136~ms & 18,104~ms & 36,456~ms & 75,437~ms\\
    \hline\hline
  \end{tabular*}
\end{table}

\begin{table}[t]
  \caption{The \cwm reasoner performs worse than EYE,
  but still creates smaller compositions acceptably fast (\emph{average times of 50~trials}).}
  \label{tbl:cwm1}
  \fontsize{8pt}{9pt}\selectfont
  \begin{tabular*}{\textwidth}{@{\extracolsep{\fill}} r r r r r r r r}
    \hline\hline
    \hspace{5.2ex}\bf \#APIs~$n$ & \bf 4 & \bf 8 & \bf 16 & \bf 32 & \bf 64 & \bf 128 & \bf 256\\
    \hline
    \bf{$d$=$1$~dep.}&&&&&&&\\
    parsing & 139~ms & 155~ms & 186~ms & 271~ms & 566~ms & 1,721~ms & 6,119~ms\\
    \em reasoning &  \em 43~ms &  \em 106~ms &  \em 320~ms &  \em 1,149~ms &  \em 4,456~ms &  \em 18,266~ms &  \em 82,856~ms\\
    total & 183~ms & 262~ms & 506~ms & 1,421~ms & 5,022~ms & 19,987~ms & 88,975~ms\\
    \hline
    \bf{$d$=$2$~dep.}&&&&&&&\\
    parsing & 143~ms & 161~ms & 201~ms & 300~ms & 637~ms & 1,904~ms & 6,509~ms\\
    \em reasoning &  \em 119~ms &  \em 473~ms &  \em 1,946~ms &  \em 8,016~ms &  \em 31,189~ms &  \em 136,806~ms &  \em 582,810~ms\\
    total & 263~ms & 634~ms & 2,147~ms & 8,317~ms & 31,827~ms & 138,710~ms & 589,319~ms\\
    \hline
    \bf{$d$=$3$~dep.}&&&&&&&\\
    parsing & 147~ms & 168~ms & 216~ms & 328~ms & 701~ms & 2,088~ms & 7,396~ms\\
    \em reasoning &  \em 190~ms &  \em 739~ms &  \em 2,843~ms &  \em 10,708~ms &  \em 43,429~ms &  \em 182,949~ms &  \em 776,569~ms\\
    total & 337~ms & 908~ms & 3,059~ms & 11,036~ms & 44,130~ms & 185,037~ms & 783,965~ms\\
    \hline\hline
  \end{tabular*}
\end{table}

\begin{table*}[t]
  \caption{The \cwm reasoner does not perform well when the number of dummies increases
           (\emph{average times of 50~trials}).}
  \label{tbl:cwm2}
  \fontsize{8pt}{9pt}\selectfont
  \begin{tabular*}{\textwidth}{@{\extracolsep{\fill}} r r r r r r r}
    \hline\hline
    \bf \#dummies~$n'$ & \bf 16 & \bf 32 & \bf 64 & \bf 128 & \bf 256 & \bf 512\\
    \hline
    \bf{$n$=32, $d$=1~dep.}&&&&&\\
    parsing & 406~ms & 576~ms & 1,057~ms & 2,548~ms & 7,682~ms & 26,678~ms\\
    \em reasoning &  \em 1,107~ms &  \em 1,083~ms &  \em 1,116~ms &  \em 2,173~ms &  \em 12,908~ms &  \em 110,775~ms\\
    total & 1,513~ms & 1,659~ms & 2,173~ms & 4,721~ms & 20,590~ms & 137,453~ms\\
    \hline\hline
  \end{tabular*}
\end{table*}

\subsubsection{Analysis}
The main cause of the difference in performance between EYE and \cwm
are due to the different reasoning mechanisms.
EYE~is a~backward-chaining reasoner,
which starts from the goal and works towards the initial state,
whereas \cwm is forward-chaining,
exploring inferences from the initial state onwards
until the goal has been reached.
This explorative behavior demands more processor time,
since all possible paths have to be tried,
even those that do not contribute to the composition.
This is most apparent in the experiment with dummies:
\cwm tries to use them and eventually finds they are not necessary;
EYE will only parse them but never tries to use them in the composition.

\vspace{-1em}
\subsubsection{Discussion}
The main question of these experiments
is whether the results are generalizable to real-world hypermedia API compositions.
On the one hand,
we have to investigate the difference between the generated API descriptions
and actual hypermedia API descriptions.
On the other hand,
we have to verify if the resulting reasoning times are acceptable
for realistic compositions in a~Web-scale environment.

First, the generated descriptions have been tailored to closely mimic actual \restdesc descriptions.
The following characteristics of actual descriptions are also found in the generated ones.
They contain a~number of pre-conditions, determined by the parameter~$d$,
and an equal number of post-conditions.
They describe the \http request that has to be performed
to execute a~hypermedia API operation.
The parameters of the request are obtained through placeholders from the pre-conditions,
so they have to be instantiated.
\\
In contrast, some characteristics are different.
All requests are~\Verb!GET! requests,
whereas real-world APIs also use other \http verbs.
However, this has no impact on the reasoner.
Furthermore, all descriptions employ predicates with a~shared \uri namespace.
This is done to ensure that a~composition always exists.
However, this too has no impact on reasoning or parsing time.
Therefore, the generated descriptions simulate real-world descriptions reliably.

Second, we note that only the reasoning times are important,
because the parsing results can be cached.
The maximum tested composition length~$n$ is large compared to
what one could expect from realistic compositions.
In practice, compositions of only a~few API operations will be necessary,
yet both reasoners perform acceptably on small to medium composition sizes.
Furthermore, EYE is capable of creating compositions of a~few hundred API operations
in just a~few hundred milliseconds.
Also, the number of dependencies~$d$ of each API will likely be limited,
with most calls only depending on a~single other operation.
Yet even if this is not the case,
EYE can fluently cope with multiple dependencies.
The final parameter we have to check is the total number of APIs~$t$,
since reasoners should be able to create compositions out of large API repositories.
Given that ProgrammableWeb contains 10,000~APIs~\cite{ProgrammableWeb},
the fact that EYE merely needs $\approx$~230~ms
to create a~composition in presence of more than 130,000~dummy APIs,
indicates that proof-based composition is a~viable strategy for the years to come.

\section{Conclusion}
\label{sec:Conclusion}
In this article,
we explained a~novel solution to automated composition and execution of hypermedia APIs.
A~crucial part in generating a~composition
is the ability to determine whether it will satisfy a~given goal
without any undesired effects.
This has led us to the approach of a~pragmatic proof,
wherein hypermedia API operations are incorporated as inference rules.
We distinguish between a~pre-execution proof and a~post-execution proof,
where the former has the additional assumption that all hypermedia API operations will succeed,
hence the ``pragmatic'' label of the method.

We selected an RDF-based method and logic for this task,
in order to bridge between existing Web technologies
and concepts from logic programming.
A~benefit of proof-based composition is that it does not require new algorithms and tools,
but can be applied with existing Semantic Web reasoners.
Those reasoners can easily incorporate external sources of knowledge
such as ontologies or business rules.
Furthermore, the performance of composition generation improves
with the evolution of those reasoners.
Also, the fact that a~third-party tool is used allows independent validation of the composition.

Our approach is a~special use case for proofs,
which have traditionally been regarded as a~part of trust on the Semantic Web.
While pre-proofs partly contribute to this,
they also have the added functionality of generating a~composition during that process.
It will be interesting to explore other opportunities
to exploit the power of proof creation
and the mechanisms behind it.
This application can serve as an example of how to apply such ideas.

In the past, we have already employed the method
in the domain of sensor APIs~\cite{verborgh_ssn_2012},
yet we want to extend the approach to other domains such as multimedia analysis and transcoding~\cite{verborgh_mtap_2013,vanlancker_mtap_2013}.
In the longterm, we aim at offering the composition method described in this article
as a~hypermedia API itself,
so it can be used for dynamic mash-up and composition~generation.

Another interesting path is to explore the limits of the used logic.
For instance, it would currently be impossible to express the deletion of resources,
even though this is a~common operation on the Web
and even has a~designated HTTP method \verb!DELETE!.
We are currently experimenting with capturing explicitly described states
inside RESTdesc descriptions to account for these situations.


A~crucial part of the proof-based method is that the interaction remains driven by hypermedia.
In contrast to traditional approaches,
where a~plan determines the full interaction,
the composition here serves as a~guideline to complete the interaction.
Until the moment machines are able to autonomously interpret
the meaning of following a~hyperlink---%
like we humans can---%
guiding them through a~hypermedia application with descriptions and proofs
can be the~pragmatic~alternative.

\subsection*{Acknowledgments}
{
  \small
  Ruben Verborgh is a~Postdoctoral Fellow of the Research Foundation Flanders.
  The research activities were funded by Ghent University,
  the Institute for the Promotion of Innovation by Science and Technology in Flanders~(IWT),
  and the European~Union.
}

\bibliographystyle{acmtrans}
\bibliography{pragmatic-proof}

\begin{thebibliography}{}

\bibitem[\protect\citeauthoryear{Abiteboul, Hull, and Vianu}{Abiteboul
  et~al\mbox{.}}{1995}]{datalog}
{\sc Abiteboul, S.}, {\sc Hull, R.}, {\sc and} {\sc Vianu, V.}, Eds. 1995.
\newblock {\em Foundations of Databases: The Logical Level\/}, 1st ed.
\newblock Addison-Wesley Longman Publishing Co., Inc., Boston, MA, USA.

\bibitem[\protect\citeauthoryear{Alarc\'{o}n and Wilde}{Alarc\'{o}n and
  Wilde}{2010}]{ReLL}
{\sc Alarc\'{o}n, R.} {\sc and} {\sc Wilde, E.} 2010.
\newblock {RESTler}: crawling {RESTful} services.
\newblock In {\em Proceedings of the 19\textsuperscript{th} international
  conference on World Wide Web}. ACM, 1051--1052.

\bibitem[\protect\citeauthoryear{Alarc\'{o}n, Wilde, and Bellido}{Alarc\'{o}n
  et~al\mbox{.}}{2011}]{ReLLComposition}
{\sc Alarc\'{o}n, R.}, {\sc Wilde, E.}, {\sc and} {\sc Bellido, J.} 2011.
\newblock Hypermedia-driven {RESTful} service composition.
\newblock In {\em Service-Oriented Computing}. Lecture Notes in Computer
  Science, vol. 6568. Springer, 111--120.

\bibitem[\protect\citeauthoryear{Angele, Boley, de~Bruijn, Fensel, Hitzler,
  Kifer, Krummenacher, Lausen, Polleres, and Studer}{Angele
  et~al\mbox{.}}{2005}]{wrl}
{\sc Angele, J.}, {\sc Boley, H.}, {\sc de~Bruijn, J.}, {\sc Fensel, D.}, {\sc
  Hitzler, P.}, {\sc Kifer, M.}, {\sc Krummenacher, R.}, {\sc Lausen, H.}, {\sc
  Polleres, A.}, {\sc and} {\sc Studer, R.} 2005.
\newblock Web rule language (wrl).
\newblock {\em W3C Member Submission\/}~{\em 9}.

\bibitem[\protect\citeauthoryear{Arndt, Verborgh, De~Roo, Sun, Mannens, and Van
  De~Walle}{Arndt et~al\mbox{.}}{2015}]{semN3}
{\sc Arndt, D.}, {\sc Verborgh, R.}, {\sc De~Roo, J.}, {\sc Sun, H.}, {\sc
  Mannens, E.}, {\sc and} {\sc Van De~Walle, R.} 2015.
\newblock Semantics of {Notation3} logic: A solution for implicit
  quantification.
\newblock In {\em Rule Technologies: Foundations, Tools, and Applications}.
  Springer, 127--143.

\bibitem[\protect\citeauthoryear{Auer, Bizer, Kobilarov, Lehmann, Cyganiak, and
  Ives}{Auer et~al\mbox{.}}{2007}]{DBpedia}
{\sc Auer, S.}, {\sc Bizer, C.}, {\sc Kobilarov, G.}, {\sc Lehmann, J.}, {\sc
  Cyganiak, R.}, {\sc and} {\sc Ives, Z.} 2007.
\newblock \dbpedia: A nucleus for a {Web} of open data.
\newblock In {\em The Semantic Web}. Lecture Notes in Computer Science, vol.
  4825. Springer Berlin Heidelberg, 722--735.
\newblock \url{http://dx.doi.org/10.1007/978-3-540-76298-0_52}.

\bibitem[\protect\citeauthoryear{Baget, Lecl{\`e}re, Mugnier, Rocher, and
  Sipieter}{Baget et~al\mbox{.}}{2015}]{graal}
{\sc Baget, J.-F.}, {\sc Lecl{\`e}re, M.}, {\sc Mugnier, M.-L.}, {\sc Rocher,
  S.}, {\sc and} {\sc Sipieter, C.} 2015.
\newblock Graal: A toolkit for query answering with existential rules.
\newblock In {\em Rule Technologies: Foundations, Tools, and Applications}.
  Springer, 328--344.

\bibitem[\protect\citeauthoryear{Baget, Leclère, Mugnier, and Salvat}{Baget
  et~al\mbox{.}}{2011}]{Baget}
{\sc Baget, J.-F.}, {\sc Leclère, M.}, {\sc Mugnier, M.-L.}, {\sc and} {\sc
  Salvat, E.} 2011.
\newblock On rules with existential variables: Walking the decidability line.
\newblock {\em Artificial Intelligence\/}~{\em 175,\/}~9--10, 1620--1654.

\bibitem[\protect\citeauthoryear{Beckett, Berners-Lee, Prud'hommeaux, and
  Carothers}{Beckett et~al\mbox{.}}{2013}]{turtle}
{\sc Beckett, D.}, {\sc Berners-Lee, T.}, {\sc Prud'hommeaux, E.}, {\sc and}
  {\sc Carothers, G.} 2013.
\newblock {Turtle - Terse RDF Triple Language}.
\newblock W3C Candidate Recommendation.
\newblock \url{http://www.w3.org/TR/turtle/}.

\bibitem[\protect\citeauthoryear{Berlind}{Berlind}{2013}]{ProgrammableWeb}
{\sc Berlind, D.} 2013.
\newblock {ProgrammableWeb's} directory hits 10,000 \apis{.} {And} counting.
\newblock ProgrammableWeb blog.
\newblock
  \url{http://blog.programmableweb.com/2013/09/23/programmablewebs-directory-hits-10000-apis-and-counting/}.

\bibitem[\protect\citeauthoryear{Berners-Lee}{Berners-Lee}{2000}]{SWAP}
{\sc Berners-Lee, T.} 2000.
\newblock {Semantic Web Application Platform}.
\newblock \url{http://www.w3.org/2000/10/swap/}.

\bibitem[\protect\citeauthoryear{Berners-Lee}{Berners-Lee}{2009}]{cwm}
{\sc Berners-Lee, T.} 2000--2009.
\newblock {\cwm}.
\newblock \url{http://www.w3.org/2000/10/swap/doc/cwm.html}.

\bibitem[\protect\citeauthoryear{Berners-Lee, Cailliau, and Groff}{Berners-Lee
  et~al\mbox{.}}{1992}]{bernerslee_1992}
{\sc Berners-Lee, T.}, {\sc Cailliau, R.}, {\sc and} {\sc Groff, J.-F.} 1992.
\newblock The world-wide web.
\newblock {\em Computer Networks and ISDN Systems\/}~{\em 25,\/}~4--5,
  454--459.

\bibitem[\protect\citeauthoryear{Berners-Lee and Connolly}{Berners-Lee and
  Connolly}{2011}]{Notation3}
{\sc Berners-Lee, T.} {\sc and} {\sc Connolly, D.} 2011.
\newblock Notation3 ({N3}): A~readable {RDF} syntax.
\newblock W3C Team Submission.
\newblock \url{http://www.w3.org/TeamSubmission/n3/}.

\bibitem[\protect\citeauthoryear{Berners-Lee, Connolly, Kagal, Scharf, and
  Hendler}{Berners-Lee et~al\mbox{.}}{2008}]{N3Logic}
{\sc Berners-Lee, T.}, {\sc Connolly, D.}, {\sc Kagal, L.}, {\sc Scharf, Y.},
  {\sc and} {\sc Hendler, J.} 2008.
\newblock {\nthreelogic}: A logical framework for the {World Wide Web}.
\newblock {\em Theory and Practice of Logic Programming\/}~{\em 8,\/}~3,
  249--269.

\bibitem[\protect\citeauthoryear{Berners-Lee, Hendler, and Lassila}{Berners-Lee
  et~al\mbox{.}}{2001}]{SemanticWeb}
{\sc Berners-Lee, T.}, {\sc Hendler, J.}, {\sc and} {\sc Lassila, O.} 2001.
\newblock The {Semantic Web}.
\newblock {\em Scientific American\/}~{\em 284,\/}~5, 34.

\bibitem[\protect\citeauthoryear{Bishop and Fischer}{Bishop and
  Fischer}{2008}]{iris}
{\sc Bishop, B.} {\sc and} {\sc Fischer, F.} 2008.
\newblock Iris-integrated rule inference system.
\newblock In {\em International Workshop on Advancing Reasoning on the Web:
  Scalability and Commonsense (ARea 2008)}.

\bibitem[\protect\citeauthoryear{Bizer, Jentzsch, and Cyganiak}{Bizer
  et~al\mbox{.}}{2011}]{LODCloud}
{\sc Bizer, C.}, {\sc Jentzsch, A.}, {\sc and} {\sc Cyganiak, R.} 2011.
\newblock State of the lod cloud.
\newblock \url{http://www4.wiwiss.fu-berlin.de/lodcloud/state}.

\bibitem[\protect\citeauthoryear{Bock, Fokoue, Haase, Hoekstra, Horrocks,
  Ruttenberg, Sattler, and Smith}{Bock et~al\mbox{.}}{2012}]{OWL}
{\sc Bock, C.}, {\sc Fokoue, A.}, {\sc Haase, P.}, {\sc Hoekstra, R.}, {\sc
  Horrocks, I.}, {\sc Ruttenberg, A.}, {\sc Sattler, U.}, {\sc and} {\sc Smith,
  M.} 2012.
\newblock Owl 2 {Web Ontology Language}.
\newblock W3C Recommendation.
\newblock \url{http://www.w3.org/TR/owl2-syntax/}.

\bibitem[\protect\citeauthoryear{Brickley and Guha}{Brickley and
  Guha}{2004}]{RDFS}
{\sc Brickley, D.} {\sc and} {\sc Guha, R.~V.} 2004.
\newblock {RDF Vocabulary Description Language 1.0: RDF Schema}.
\newblock W3C Recommendation.
\newblock \url{http://www.w3.org/TR/rdf-schema/}.

\bibitem[\protect\citeauthoryear{Cal{\`\i}, Gottlob, Lukasiewicz, and
  Pieris}{Cal{\`\i} et~al\mbox{.}}{2011}]{datalogpm}
{\sc Cal{\`\i}, A.}, {\sc Gottlob, G.}, {\sc Lukasiewicz, T.}, {\sc and} {\sc
  Pieris, A.} 2011.
\newblock Datalog+/-: A family of languages for ontology querying.
\newblock In {\em Datalog Reloaded}. Springer, 351--368.

\bibitem[\protect\citeauthoryear{Carroll, Dickinson, Dollin, Reynolds,
  Seaborne, and Wilkinson}{Carroll et~al\mbox{.}}{2004}]{Jena}
{\sc Carroll, J.}, {\sc Dickinson, I.}, {\sc Dollin, C.}, {\sc Reynolds, D.},
  {\sc Seaborne, A.}, {\sc and} {\sc Wilkinson, K.} 2004.
\newblock Jena: implementing the {Semantic Web} recommendations.
\newblock In {\em Proceedings of the 13\textsuperscript{th} international World
  Wide Web conference}. ACM, 74--83.
\newblock \url{www.hpl.hp.com/techreports/2003/HPL-2003-146.pdf}.

\bibitem[\protect\citeauthoryear{Christensen, Curbera, Meredith, and
  Weerawarana}{Christensen et~al\mbox{.}}{2001}]{WSDL}
{\sc Christensen, E.}, {\sc Curbera, F.}, {\sc Meredith, G.}, {\sc and} {\sc
  Weerawarana, S.} 2001.
\newblock {Web Services Description Language (WSDL)}.
\newblock W3C Note.
\newblock \url{http://www.w3.org/TR/wsdl}.

\bibitem[\protect\citeauthoryear{Clocksin and Mellish}{Clocksin and
  Mellish}{1994}]{Prolog}
{\sc Clocksin, W.~F.} {\sc and} {\sc Mellish, C.~S.} 1994.
\newblock {\em Programming in {PROLOG}}.
\newblock Springer.

\bibitem[\protect\citeauthoryear{De~Roo}{De~Roo}{2014}]{EYE}
{\sc De~Roo, J.} 1999--2014.
\newblock Euler proof mechanism.
\newblock \url{http://eulersharp.sourceforge.net/}.

\bibitem[\protect\citeauthoryear{Duerst and Suignard}{Duerst and
  Suignard}{2005}]{iri}
{\sc Duerst, M.} {\sc and} {\sc Suignard, M.} 2005.
\newblock {Internationalized Resource Identifiers (IRIs)}.
\newblock \url{http://www.ietf.org/rfc/rfc3987.txt}.

\bibitem[\protect\citeauthoryear{Fielding}{Fielding}{2008}]{fielding_2008}
{\sc Fielding, R.~T.} 2008.
\newblock {REST APIs} must be hypertext-driven.
\newblock Untangled -- Musings of Roy T. Fielding.
\newblock
  \url{http://roy.gbiv.com/untangled/2008/rest-apis-must-be-hypertext-driven}.

\bibitem[\protect\citeauthoryear{Fielding, Gettys, Mogul, Frystyk, Masinter,
  Leach, and Berners-Lee}{Fielding et~al\mbox{.}}{1999}]{HTTP}
{\sc Fielding, R.~T.}, {\sc Gettys, J.}, {\sc Mogul, J.}, {\sc Frystyk, H.},
  {\sc Masinter, L.}, {\sc Leach, P.}, {\sc and} {\sc Berners-Lee, T.} 1999.
\newblock {Hypertext Transfer Protocol -- HTTP/1.1}.
\newblock \url{http://www.ietf.org/rfc/rfc2616.txt}.

\bibitem[\protect\citeauthoryear{Fielding and Taylor}{Fielding and
  Taylor}{2002}]{REST}
{\sc Fielding, R.~T.} {\sc and} {\sc Taylor, R.~N.} 2002.
\newblock {Principled design of the modern Web architecture}.
\newblock {\em Transactions on Internet Technology\/}~{\em 2,\/}~2 (May),
  115--150.

\bibitem[\protect\citeauthoryear{Gomadam, Ranabahu, and Sheth}{Gomadam
  et~al\mbox{.}}{2010}]{SAREST}
{\sc Gomadam, K.}, {\sc Ranabahu, A.}, {\sc and} {\sc Sheth, A.} 2010.
\newblock {SA-REST: Semantic Annotation of Web Resources}.
\newblock W3C Member Submission.
\newblock \url{http://www.w3.org/Submission/SA-REST/}.

\bibitem[\protect\citeauthoryear{Gottlob, Orsi, and Pieris}{Gottlob
  et~al\mbox{.}}{2014}]{irispm}
{\sc Gottlob, G.}, {\sc Orsi, G.}, {\sc and} {\sc Pieris, A.} 2014.
\newblock Query rewriting and optimization for ontological databases.
\newblock {\em ACM Transactions on Database Systems (TODS)\/}~{\em 39,\/}~3,
  25.

\bibitem[\protect\citeauthoryear{Hayes and Patel-Schneider}{Hayes and
  Patel-Schneider}{2014}]{RDFSemantics}
{\sc Hayes, P.~J.} {\sc and} {\sc Patel-Schneider, P.~F.} 2014.
\newblock {{\sc rdf} 1.1 Semantics}.
\newblock {\sc w\oldstylenums3c} Recommendation.
\newblock \url{http://www.w3.org/TR/2014/REC-rdf11-mt-20140225/}.

\bibitem[\protect\citeauthoryear{Horrocks, Patel-Schneider, Boley, Tabet,
  Grosof, and Dean}{Horrocks et~al\mbox{.}}{2004}]{swrl}
{\sc Horrocks, I.}, {\sc Patel-Schneider, P.~F.}, {\sc Boley, H.}, {\sc Tabet,
  S.}, {\sc Grosof, B.}, {\sc and} {\sc Dean, M.} 2004.
\newblock {SWRL}: A semantic web rule language combining {OWL} and {RuleML}.
\newblock {W3C} {M}ember {S}ubmission.
\newblock Available at \url{http://www.w3.org/Submission/SWRL/}.

\bibitem[\protect\citeauthoryear{Kifer}{Kifer}{2008}]{rif}
{\sc Kifer, M.} 2008.
\newblock Rule interchange format: The framework.
\newblock In {\em Web reasoning and rule systems}. Springer, 1--11.

\bibitem[\protect\citeauthoryear{Kifer, Lausen, and Wu}{Kifer
  et~al\mbox{.}}{1995}]{flogic}
{\sc Kifer, M.}, {\sc Lausen, G.}, {\sc and} {\sc Wu, J.} 1995.
\newblock Logical foundations of object-oriented and frame-based languages.
\newblock {\em Journal of the ACM (JACM)\/}~{\em 42,\/}~4, 741--843.

\bibitem[\protect\citeauthoryear{Klyne and Carrol}{Klyne and
  Carrol}{2004}]{RDF}
{\sc Klyne, G.} {\sc and} {\sc Carrol, J.~J.} 2004.
\newblock {Resource Description Framework (RDF): Concepts and Abstract Syntax}.
\newblock W3C Recommendation.
\newblock \url{http://www.w3.org/TR/2004/REC-rdf-concepts-20040210/}.

\bibitem[\protect\citeauthoryear{Koch, Valesco, and Ackermann}{Koch
  et~al\mbox{.}}{2011}]{httprdf}
{\sc Koch, J.}, {\sc Valesco, C.~A.}, {\sc and} {\sc Ackermann, P.} 2011.
\newblock {HTTP} vocabulary in {RDF} 1.0.
\newblock W3C Working Draft.
\newblock \url{http://www.w3.org/TR/HTTP-in-RDF10/}.

\bibitem[\protect\citeauthoryear{Kopeck\'{y}, Gomadam, and Vitvar}{Kopeck\'{y}
  et~al\mbox{.}}{2008}]{hRESTS}
{\sc Kopeck\'{y}, J.}, {\sc Gomadam, K.}, {\sc and} {\sc Vitvar, T.} 2008.
\newblock {hRESTS}: An {HTML} microformat for describing {RESTful} {Web}
  services.
\newblock In {\em Proceedings of the International Conference on Web
  Intelligence and Intelligent Agent Technology}. IEEE Computer Society,
  619--625.

\bibitem[\protect\citeauthoryear{Kopeck\'{y} and Vitvar}{Kopeck\'{y} and
  Vitvar}{2008}]{MicroWSMO}
{\sc Kopeck\'{y}, J.} {\sc and} {\sc Vitvar, T.} 2008.
\newblock Microwsmo.
\newblock WSMO Working Draft.
\newblock \url{http://www.wsmo.org/TR/d38/v0.1/}.

\bibitem[\protect\citeauthoryear{Kopeck\'y, Vitvar, Bournez, and
  Farrell}{Kopeck\'y et~al\mbox{.}}{2007}]{SAWSDL}
{\sc Kopeck\'y, J.}, {\sc Vitvar, T.}, {\sc Bournez, C.}, {\sc and} {\sc
  Farrell, J.} 2007.
\newblock {Semantic Annotations for WSDL and XML Schema}.
\newblock {\em IEEE Internet Computing\/}~{\em 11}, 60--67.

\bibitem[\protect\citeauthoryear{Lanthaler and G\"utl}{Lanthaler and
  G\"utl}{2013}]{HydraVocabulary}
{\sc Lanthaler, M.} {\sc and} {\sc G\"utl, C.} 2013.
\newblock Hydra: A vocabulary for hypermedia-driven {Web} \apis.
\newblock In {\em Proceedings of the 6\textsuperscript{th}~Workshop on Linked
  Data on the Web}.

\bibitem[\protect\citeauthoryear{Lausen, Polleres, and Roman}{Lausen
  et~al\mbox{.}}{2005}]{WSMO}
{\sc Lausen, H.}, {\sc Polleres, A.}, {\sc and} {\sc Roman, D.} 2005.
\newblock {Web Service Modeling Ontology (WSMO)}.
\newblock W3C Member Submission.
\newblock \url{http://www.w3.org/Submission/WSMO/}.

\bibitem[\protect\citeauthoryear{Lloyd and Topor}{Lloyd and
  Topor}{1984}]{sugar}
{\sc Lloyd, J.~W.} {\sc and} {\sc Topor, R.~W.} 1984.
\newblock Making prolog more expressive.
\newblock {\em The Journal of Logic Programming\/}~{\em 1,\/}~3, 225--240.

\bibitem[\protect\citeauthoryear{Maleshkova, Kopeck\'{y}, and
  Pedrinaci}{Maleshkova et~al\mbox{.}}{2009}]{Maleshkova2009}
{\sc Maleshkova, M.}, {\sc Kopeck\'{y}, J.}, {\sc and} {\sc Pedrinaci, C.}
  2009.
\newblock Adapting {SAWSDL} for semantic annotations of {RESTful} services.
\newblock In {\em Proceedings of the On the Move to Meaningful Internet Systems
  Workshops}. Lecture Notes in Computer Science, vol. 5872. Springer, 917--926.

\bibitem[\protect\citeauthoryear{Manna and Waldinger}{Manna and
  Waldinger}{1980}]{Manna}
{\sc Manna, Z.} {\sc and} {\sc Waldinger, R.} 1980.
\newblock A deductive approach to program synthesis.
\newblock {\em Transactions on Programming Languages and Systems
  (TOPLAS)\/}~{\em 2,\/}~1, 90--121.

\bibitem[\protect\citeauthoryear{Martin, Burstein, Hobbs, and Lassila}{Martin
  et~al\mbox{.}}{2004}]{OWLS}
{\sc Martin, D.}, {\sc Burstein, M.}, {\sc Hobbs, J.}, {\sc and} {\sc Lassila,
  O.} 2004.
\newblock {\owls: Semantic Markup for Web Services}.
\newblock W3C Member Submission.
\newblock \url{http://www.w3.org/Submission/OWL-S/}.

\bibitem[\protect\citeauthoryear{Milanovic and Malek}{Milanovic and
  Malek}{2004}]{CurrentComposition}
{\sc Milanovic, N.} {\sc and} {\sc Malek, M.} 2004.
\newblock Current solutions for {Web} service composition.
\newblock {\em Internet Computing, IEEE\/}~{\em 8,\/}~6, 51--59.

\bibitem[\protect\citeauthoryear{Mugnier}{Mugnier}{2011}]{exrules}
{\sc Mugnier, M.-L.} 2011.
\newblock Ontological query answering with existential rules.
\newblock In {\em Web Reasoning and Rule Systems}. Springer, 2--23.

\bibitem[\protect\citeauthoryear{Norton and Krummenacher}{Norton and
  Krummenacher}{2010}]{LOS}
{\sc Norton, B.} {\sc and} {\sc Krummenacher, R.} 2010.
\newblock {Consuming dynamic Linked Data}.
\newblock In {\em Proceedings of the 1\textsuperscript{st} International
  Workshop on Consuming Linked Data}.

\bibitem[\protect\citeauthoryear{Ordóñez-Ante, Rojas-Potosi, Suarez-Meza, and
  Corrales}{Ordóñez-Ante et~al\mbox{.}}{2012}]{Leandro}
{\sc Ordóñez-Ante, L.}, {\sc Rojas-Potosi, L.~A.}, {\sc Suarez-Meza, L.~J.},
  {\sc and} {\sc Corrales, J.~C.} 2012.
\newblock Towards the automation of the semantic annotation process for {Web}
  services.
\newblock In {\em Proceedings of the 2012 Conference on Semantic Web \& Web
  Services}.

\bibitem[\protect\citeauthoryear{Parsia and Sirin}{Parsia and
  Sirin}{2004}]{Pellet}
{\sc Parsia, B.} {\sc and} {\sc Sirin, E.} 2004.
\newblock Pellet: An {OWL DL} reasoner.
\newblock In {\em Proceedings of the Third International Semantic Web
  Conference}.

\bibitem[\protect\citeauthoryear{Pautasso and Wilde}{Pautasso and
  Wilde}{2009}]{WebLooselyCoupled}
{\sc Pautasso, C.} {\sc and} {\sc Wilde, E.} 2009.
\newblock Why is the {Web} loosely coupled? -- {A} multi-faceted metric for
  service design.
\newblock In {\em Proceedings of the 18\textsuperscript{th}~International
  Conference on World Wide Web}. ACM, New York, 911--920.

\bibitem[\protect\citeauthoryear{Richardson, Amundsen, and Ruby}{Richardson
  et~al\mbox{.}}{2013}]{RESTfulWebApis}
{\sc Richardson, L.}, {\sc Amundsen, M.}, {\sc and} {\sc Ruby, S.} 2013.
\newblock {\em {RESTful} {Web} {\apis}}.
\newblock O'Reilly.

\bibitem[\protect\citeauthoryear{Speiser and Harth}{Speiser and
  Harth}{2011}]{LIDS}
{\sc Speiser, S.} {\sc and} {\sc Harth, A.} 2011.
\newblock {Integrating Linked Data and Services with Linked Data Services}.
\newblock In {\em The Semantic Web: Research and Applications}. Lecture Notes
  in Computer Science, vol. 6643. Springer, 170--184.

\bibitem[\protect\citeauthoryear{Stadtm{\"u}ller, Speiser, Harth, and
  Studer}{Stadtm{\"u}ller et~al\mbox{.}}{2013}]{DataFu}
{\sc Stadtm{\"u}ller, S.}, {\sc Speiser, S.}, {\sc Harth, A.}, {\sc and} {\sc
  Studer, R.} 2013.
\newblock {Data-Fu}: a language and an interpreter for interaction with
  read/write {Linked Data}.
\newblock In {\em Proceedings of the World Wide Web Conference}. 1225--1236.

\bibitem[\protect\citeauthoryear{Van~Lancker, Van~Deursen, Verborgh, and Van~de
  Walle}{Van~Lancker et~al\mbox{.}}{2013}]{vanlancker_mtap_2013}
{\sc Van~Lancker, W.}, {\sc Van~Deursen, D.}, {\sc Verborgh, R.}, {\sc and}
  {\sc Van~de Walle, R.} 2013.
\newblock Semantic media decision taking using {N3Logic}.
\newblock {\em Multimedia Tools and Applications\/}~{\em 63,\/}~1 (Mar.),
  7--26.

\bibitem[\protect\citeauthoryear{Verborgh and De~Roo}{Verborgh and
  De~Roo}{2015}]{eyepaper}
{\sc Verborgh, R.} {\sc and} {\sc De~Roo, J.} 2015.
\newblock Drawing conclusions from linked data on the web.
\newblock {\em IEEE Software\/}~{\em 32,\/}~5 (May).

\bibitem[\protect\citeauthoryear{Verborgh, Haerinck, Steiner, Van~Deursen,
  Van~Hoecke, De~Roo, Van~de Walle, and Gabarr\'o~Vall\'es}{Verborgh
  et~al\mbox{.}}{2012}]{verborgh_ssn_2012}
{\sc Verborgh, R.}, {\sc Haerinck, V.}, {\sc Steiner, T.}, {\sc Van~Deursen,
  D.}, {\sc Van~Hoecke, S.}, {\sc De~Roo, J.}, {\sc Van~de Walle, R.}, {\sc
  and} {\sc Gabarr\'o~Vall\'es, J.} 2012.
\newblock Functional composition of sensor {Web APIs}.
\newblock In {\em Proceedings of the 5th International Workshop on Semantic
  Sensor Networks}.

\bibitem[\protect\citeauthoryear{Verborgh, Harth, Maleshkova, Stadtm\"uller,
  Steiner, Taheriyan, and Van~de Walle}{Verborgh
  et~al\mbox{.}}{2014}]{verborgh_rest_2014}
{\sc Verborgh, R.}, {\sc Harth, A.}, {\sc Maleshkova, M.}, {\sc Stadtm\"uller,
  S.}, {\sc Steiner, T.}, {\sc Taheriyan, M.}, {\sc and} {\sc Van~de Walle, R.}
  2014.
\newblock Survey of semantic description of {REST APIs}.
\newblock In {\em REST: Advanced Research Topics and Practical Applications},
  {C.~Pautasso}, {E.~Wilde}, {and} {R.~Alarc\'on}, Eds. Springer, 69--89.

\bibitem[\protect\citeauthoryear{Verborgh, Steiner, Van~Deursen, Coppens,
  Gabarr\'o~Vall\'es, and Van~de Walle}{Verborgh
  et~al\mbox{.}}{2012}]{verborgh_wsrest_2012}
{\sc Verborgh, R.}, {\sc Steiner, T.}, {\sc Van~Deursen, D.}, {\sc Coppens,
  S.}, {\sc Gabarr\'o~Vall\'es, J.}, {\sc and} {\sc Van~de Walle, R.} 2012.
\newblock Functional descriptions as the bridge between hypermedia apis and the
  {Semantic Web}.
\newblock In {\em Proceedings of the Third International Workshop on RESTful
  Design}. ACM, 33--40.

\bibitem[\protect\citeauthoryear{Verborgh, Steiner, Van~Deursen, De~Roo, Van~de
  Walle, and Gabarr\'o~Vall\'es}{Verborgh
  et~al\mbox{.}}{2013}]{verborgh_mtap_2013}
{\sc Verborgh, R.}, {\sc Steiner, T.}, {\sc Van~Deursen, D.}, {\sc De~Roo, J.},
  {\sc Van~de Walle, R.}, {\sc and} {\sc Gabarr\'o~Vall\'es, J.} 2013.
\newblock {Capturing the functionality of Web services with functional
  descriptions}.
\newblock {\em Multimedia Tools and Applications\/}.

\bibitem[\protect\citeauthoryear{Verborgh, van Hooland, Cope, Chan, Mannens,
  and Van~de Walle}{Verborgh et~al\mbox{.}}{2015}]{verborgh_jod_2014}
{\sc Verborgh, R.}, {\sc van Hooland, S.}, {\sc Cope, A.~S.}, {\sc Chan, S.},
  {\sc Mannens, E.}, {\sc and} {\sc Van~de Walle, R.} 2015.
\newblock The fallacy of the multi-{API} culture: Conceptual and practical
  benefits of representational state transfer ({REST}).
\newblock {\em Journal of Documentation\/}~{\em 71,\/}~2 (Mar.).

\bibitem[\protect\citeauthoryear{Waldinger}{Waldinger}{2001}]{Waldinger}
{\sc Waldinger, R.} 2001.
\newblock Web agents cooperating deductively.
\newblock In {\em Formal Approaches to Agent-Based Systems}. Lecture Notes in
  Computer Science, vol. 1871. Springer, 250--262.

\end{thebibliography}

\end{document}